\documentclass[12pt,letterpaper]{article}

\usepackage[left=1in,right=1in,top=1in,bottom=1in]{geometry}
\usepackage{setspace}

\usepackage{amsmath,amssymb,amsthm,mathtools,bm}

\usepackage{graphicx}
\usepackage{adjustbox}
\usepackage{booktabs}
\usepackage{threeparttable}
\usepackage{longtable}
\usepackage{array}
\usepackage{multirow}
\usepackage{caption}
\usepackage{subcaption}
\usepackage{float}
\usepackage{placeins}
\usepackage{makecell}

\usepackage{enumitem}
\usepackage{xcolor}
\usepackage{appendix}
\usepackage{pdflscape}
\usepackage{csquotes}
\usepackage{bbm}

\usepackage{chngcntr}
\counterwithin{figure}{section}
\counterwithin{table}{section}

\usepackage[authoryear]{natbib}

\usepackage[colorlinks=true,linkcolor=blue,citecolor=blue,urlcolor=blue]{hyperref}
\usepackage[nameinlink,capitalize]{cleveref}

\usepackage{microtype}                       

\newtheorem{proposition}{Proposition}
\newtheorem{assumption}{Assumption}

\theoremstyle{definition}

\title{A Cap--Axis Integral Diagnostic of Factor Models}

\author{Useong Shin\thanks{
		Sogang Business School, Sogang University (Seoul, Korea).\\
		ORCID: \href{https://orcid.org/0009-0003-0197-9003}{0009-0003-0197-9003}\\
		Email: \texttt{useong@sogang.ac.kr}
}}
\date{\today}

\begin{document}
	
	\maketitle
	\thispagestyle{empty}
	
	\begin{flushleft}
	\end{flushleft}
	
	
	\begin{abstract}
		I propose a cap-axis zero-alpha diagnostic for factor-model evaluation. Whole-stock capitalization prefixes are paired with equal realized exposure to the aggregate market, producing a bridge-alpha curve that localizes pricing errors within the market. Finite-grid HAC-Gaussian inference and residual-block calibration provide size-controlled functional tests. In 1967--2024 CRSP data, q5's negative daily bridge attenuates under lead--lag correction and is small monthly, whereas Fama--French and Carhart bridges become more visible monthly. Across 155 factors, cap-axis magnitude is neither a monotone transformation of maximum-Sharpe gain nor explained by exposure to FF3 SMB.
	\end{abstract}
	
	\pagenumbering{arabic}
	
	\newpage
	\section{Introduction}
	\label{sec:intro}
	
	Factor models are commonly evaluated by whether they expand the attainable mean--variance frontier and whether they leave zero pricing errors on chosen test assets. These criteria are closely connected for an exact linear pricing model but can diverge for low-dimensional approximations, where a factor can improve the maximum Sharpe ratio while leaving systematic alphas on an economically meaningful return space.
	
	I make this separation measurable inside the market portfolio. The proposed \emph{cap-axis zero-alpha diagnostic} orders stocks by market capitalization and, at each cumulative-cap target \(p\), pairs a whole-stock prefix with an equal realized exposure to the aggregate market. The resulting zero-investment bridge \(D_t(p)\) closes at \(p=0\) and \(p=1\); regressing it on model \(m\) as \(p\) varies produces the bridge-alpha curve \(\alpha_m(p)\).
	
	The population null is zero alpha on every implemented bridge. With a zero aggregate-market alpha, linearity extends this restriction to the finite return space generated by the market and the bridge family (\Cref{prop:cap_axis_sufficiency}). Because whole-stock prefix exposure varies over time, this is not equivalent to pricing every raw size portfolio and does not establish full SDF validity: it establishes zero-alpha consistency on a fixed family of equal-exposure prefix-versus-market trades.
	
	I summarize the curve using signed area \(SA\), integrated absolute error \(IAE\), integrated squared error \(ISE\), and maximum absolute error \(SUP\). A rank-area portfolio provides a scalar directional counterpart to \(SA\), while the nonlinear functionals are tested using a finite-grid HAC-Gaussian approximation. Residual-block calibration produces near-nominal rejection under the zero-curve null and increasing power against structured alternatives.
	
	In the 1967--2024 CRSP market, all six candidate models closely price the aggregate market but leave different cap-axis footprints. q5 has a one-sided negative daily bridge that attenuates under lead--lag correction and is small monthly, whereas the Fama--French and Carhart models leave positive, nearly one-signed monthly curves. The diagnostic therefore identifies model-specific and horizon-dependent pricing errors rather than producing a universal ranking.
	
	Finally, I apply the diagnostic to 155 monthly factors added individually to the market (\Cref{sec:coordinate}). Cap-axis magnitude is neither a monotone transformation of maximum-Sharpe gain nor explained by exposure to FF3 SMB. The diagnostic thus complements joint alpha and spanning tests by recovering the direction, magnitude, location, and frequency of pricing errors along a predetermined market-internal coordinate.
	
	\section{Theoretical Background and Related Literature}
	\label{sec:lit}
	
	\paragraph{Model comparison criteria.}
	The expansion of multifactor models has made model comparison a central issue in empirical asset pricing. Standard candidates include the Fama--French models \citep{FF92,FF93,FF15}, momentum \citep{JT93,Carhart97}, and the investment-and-profitability \(q\)-factor family and its anomaly-replication applications \citep{HXZ15,HMXZ19,HXZ20,HMXZ21,HMXZ24}. A leading comparison criterion is factor spanning and the associated maximum-Sharpe-ratio frontier \citep{BS17,BS18,FF18}: a model is evaluated by whether its factor span improves the attainable mean--variance opportunity set.
	
	This paper studies a complementary object rather than replacing that criterion. A factor span can expand the mean--variance frontier while leaving pricing errors on a fixed economic subspace; the cap-axis diagnostic measures such errors on one subspace, the market portfolio's internal capitalization-rank axis.
	
	\paragraph{Test-asset dependence.}
	Joint alpha tests and distance measures depend on the chosen test assets \citep{GRS89,HJ97,Cochrane05}, so test-portfolio construction can affect inference and rankings \citep{LNS10}, especially with weak factors, redundant factors, or repeated model search \citep{GXZ25,KNS18,LM90}. Maximum-Sharpe comparisons reduce this dependence by focusing on the factor span. The cap-axis diagnostic takes the opposite route: it keeps the zero-alpha question but fixes the test-asset coordinate. The test assets are zero-investment bridges comparing each cumulative capitalization-rank prefix with an equal-exposure aggregate market, so the coordinate is generated by the market portfolio's own capitalization weights, not by an anomaly sort chosen ex post.
	
	\paragraph{Characteristic-sorted portfolios as functions.}
	The diagnostic is closest in spirit to work that treats characteristic-sorted portfolio returns as functions over a characteristic domain. \citet{CCFS20} develop estimation and inference for characteristic-sorted portfolios; I apply this functional view to market-capitalization rank and to a specific object, the bridge-alpha curve. Instead of testing only a finite set of pointwise alphas, I summarize the curve by a signed area and by \(IAE\), \(ISE\), and \(SUP\) norms, which separate direction, total magnitude, and local extremity of pricing errors along the cap-rank coordinate.
	
	\paragraph{Factor coordinates.}
	The factor-coordinate exercise uses the replicated anomaly factors of \citet{JKP23}, distributed through Global Factor Data \citep{GFD}, augmented with the Fama--French/Carhart and \(q\)-factor series used elsewhere in the paper. This broad factor set allows the cap-axis footprint to be compared with familiar coordinates such as maximum-Sharpe gain and size exposure. The exercise is descriptive: it asks whether the cap-axis norm is a restatement of existing factor coordinates or a distinct feature of factor behavior.
	
	\paragraph{Frequency and non-synchronous trading.}
	Daily factor regressions can be affected by non-synchronous trading: if small or illiquid stocks react to market-wide information with a lag, daily beta estimates can be attenuated and intercepts distorted. The lead--lag corrections of \citet{SW77} and \citet{Dimson79} address this by aggregating lagged and leading factor loadings, and I apply the same logic to the cap-axis bridge. Varying the lead--lag window produces a horizon profile of the signed-area alpha and helps distinguish short-horizon attenuation from lower-frequency pricing-error structure.
	
	\section{The Cap--Axis Zero-Alpha Diagnostic}
	\label{sec:cap_axis_integral_definition}
	
	This section defines the cap-axis diagnostic as a family of implementable zero-investment returns. The empirical object is finite-dimensional: stocks are ordered by market capitalization at a formation date, their portfolio weights evolve between formation dates, and each point on the axis compares a whole-stock prefix with an equal-exposure position in the aggregate market. A continuous cap-share representation aids interpretation but is not the object used to construct the test returns.
	
	The diagnostic is deliberately restricted: it does not test whether a factor model prices every asset or defines a valid stochastic discount factor on the full return space, but whether the model leaves zero-alpha violations on the bridge-return family generated by moving through the market portfolio in capitalization order.
	
	\subsection{The implemented finite-market bridge}
	\label{subsec:discrete_integral_definition}
	
	Let \(\tau(t)\) denote the most recent formation date before return date \(t\). At \(\tau(t)\), eligible stocks are sorted in descending order of formation-date market capitalization, and the ordering is held fixed until the next formation date. Index the stocks in this order by \(i=1,\ldots,N_t\).
	
	Let \(H_{i,t|\tau(t)}\) be the beginning-of-period value of stock \(i\)'s position under the formation-and-holding rule. Define its normalized market weight by
	\begin{equation}
		w_{i,t}
		=
		\frac{H_{i,t|\tau(t)}}{\sum_{j=1}^{N_t}H_{j,t|\tau(t)}},
		\qquad
		\sum_{i=1}^{N_t}w_{i,t}=1.
		\label{eq:implemented_market_weight}
	\end{equation}
	The weights need not equal their formation-date values: under a buy-and-hold design they evolve with the underlying positions even though the capitalization order remains fixed. The details of this holding-value evolution are specified in \Cref{sec:data}.
	
	Let \(R_{i,t}\) be the stock's total return. The return on the corresponding aggregate-market portfolio is
	\begin{equation}
		R^M_t
		=
		\sum_{i=1}^{N_t}w_{i,t}R_{i,t}.
		\label{eq:discrete_market_return}
	\end{equation}
	
	For the ordered portfolio, define the cumulative weight after stock \(k\) as
	\begin{equation}
		c_{k,t}
		=
		\sum_{i=1}^{k}w_{i,t},
		\qquad
		c_{0,t}=0.
		\label{eq:discrete_cumulative_share}
	\end{equation}
	For a target cap-share cutoff \(p\in[0,1]\), the implementation does not split the boundary stock but includes the smallest whole-stock prefix whose cumulative weight reaches \(p\):
	\begin{equation}
		k_t(p)
		=
		\begin{cases}
			0, & p=0,\\[2pt]
			\min\{k:c_{k,t}\geq p\}, & p\in(0,1],
		\end{cases}
		\qquad
		s_t(p)
		=
		c_{k_t(p),t}.
		\label{eq:discrete_cut_and_share}
	\end{equation}
	Here \(p\) is the target grid location and \(s_t(p)\) the realized market exposure of the whole-stock prefix; the two generally differ because the boundary stock is not fractionally assigned.
	
	The return contribution of the prefix is
	\begin{equation}
		C_t(p)
		=
		\sum_{i=1}^{k_t(p)}w_{i,t}R_{i,t}.
		\label{eq:discrete_prefix_contribution}
	\end{equation}
	The cap-axis bridge return is
	\begin{equation}
		D_t(p)
		=
		C_t(p)-s_t(p)R^M_t.
		\label{eq:discrete_bridge_return}
	\end{equation}
	This return longs the whole-stock prefix and shorts the aggregate market at the same realized exposure \(s_t(p)\); it is therefore zero-investment and requires no separate risk-free-rate subtraction.
	
	The bridge closes exactly at both endpoints:
	\begin{equation}
		D_t(0)=D_t(1)=0.
		\label{eq:bridge_endpoint_condition}
	\end{equation}
	To see the body--tail relation, define the complementary tail contribution \(C_t^T(p)=R^M_t-C_t(p)\). Its equal-exposure bridge is
	\begin{equation}
		D_t^T(p)
		=
		C_t^T(p)-\{1-s_t(p)\}R^M_t
		=
		-D_t(p).
		\label{eq:continuous_tail_bridge_identity}
	\end{equation}
	Thus one prefix bridge records the body--tail offset at every implemented cutoff, an identity that is exact despite whole-stock assignment and time variation in the realized prefix share.
	
	For factor model \(m\), let \(f_{m,t}\) denote its vector of factor returns. At each grid point \(p_l\), estimate
	\begin{equation}
		D_t(p_l)
		=
		\alpha_m(p_l)
		+
		\beta_m(p_l)'f_{m,t}
		+
		\varepsilon_{m,t}(p_l).
		\label{eq:discrete_bridge_regression}
	\end{equation}
	The vector
	\begin{equation}
		\boldsymbol{\alpha}_m
		=
		\left(
		\alpha_m(p_1),\ldots,\alpha_m(p_L)
		\right)'
		\label{eq:grid_alpha_vector}
	\end{equation}
	is the model's cap-axis bridge-alpha path on the implemented grid. The cap-axis zero-alpha null is
	\begin{equation}
		H^{cap}_{0,m}:
		\boldsymbol{\alpha}_m=\boldsymbol{0}.
		\label{eq:pointwise_bridge_null}
	\end{equation}
	
	\subsection{A fractional benchmark and the rank-area identity}
	\label{subsec:continuous_integral_definition}
	\label{subsec:rank_area_identity}
	
	A fractional-boundary representation provides a useful interpretation of the finite bridge. On date \(t\), assign stock \(i\) to the cap-share interval \((c_{i-1,t},c_{i,t}]\) and define the step return function \(r_t(u)=R_{i,t}\) for \(u\in(c_{i-1,t},c_{i,t}]\). The coordinate \(u\in[0,1]\) measures current market exposure while preserving the capitalization order fixed at \(\tau(t)\), so the aggregate-market return can be written as
	\begin{equation}
		R^M_t
		=
		\int_0^1 r_t(u)\,du.
		\label{eq:continuous_market_return}
	\end{equation}
	
	If the boundary stock could be split fractionally, the prefix contribution at exactly \(p\) would be
	\begin{equation}
		C_t^\ast(p)
		=
		\int_0^p r_t(u)\,du,
		\label{eq:continuous_prefix_contribution}
	\end{equation}
	and the corresponding bridge would be
	\begin{equation}
		D_t^\ast(p)
		=
		C_t^\ast(p)-pR^M_t.
		\label{eq:continuous_bridge_return}
	\end{equation}
	Unlike the whole-stock implementation, this benchmark has the constant exposure \(p\) at every date. Integrating the fractional bridge across the axis gives
	\begin{align}
		\int_0^1D_t^\ast(p)\,dp
		&=
		\int_0^1
		\left\{
		\int_0^p r_t(u)\,du-pR^M_t
		\right\}dp \nonumber\\
		&=
		\int_0^1
		\left(\frac{1}{2}-u\right)r_t(u)\,du.
		\label{eq:rank_area_continuous_identity}
	\end{align}
	
	For the finite portfolio, define stock \(i\)'s midpoint cap share as \(\mu_{i,t}=c_{i,t}-\tfrac{1}{2}w_{i,t}\). The rank-area return is
	\begin{equation}
		Y_t
		=
		\sum_{i=1}^{N_t}
		w_{i,t}R_{i,t}
		\left(
		\frac{1}{2}-\mu_{i,t}
		\right).
		\label{eq:discrete_rank_area_return}
	\end{equation}
	Because \(\sum_{i=1}^{N_t}w_{i,t}\mu_{i,t}=\tfrac{1}{2}\), the portfolio weights in \(Y_t\) sum to zero, so the rank-area return is also zero-investment.
	
	For the fractional-boundary benchmark, the step-function construction implies the exact identity
	\begin{equation}
		Y_t
		=
		\int_0^1D_t^\ast(p)\,dp.
		\label{eq:rank_area_equals_bridge_area}
	\end{equation}
	Since the OLS intercept is linear in the dependent return,
	\begin{equation}
		\alpha_m(Y)
		=
		\int_0^1\alpha_m^\ast(p)\,dp.
		\label{eq:rank_alpha_equals_signed_area}
	\end{equation}
	The equality is exact for the fractional bridge but not for the implemented whole-stock grid: the empirical bridge uses the realized share \(s_t(p)\) and does not split the boundary stock, whereas \(Y_t\) uses each stock's midpoint share. The alpha of \(Y_t\) is therefore a scalar counterpart to, rather than an algebraically identical implementation of, the grid-integrated signed area. I report the difference between the two as a boundary-quantization audit.
	
	\subsection{Functional summaries of the bridge-alpha path}
	\label{subsec:bridge_functionals}
	
	Let \(\mathcal{P}=\{p_1,\ldots,p_L\}\) with \(0=p_1<\cdots<p_L=1\) denote the empirical grid, and let \(a_l\) be its trapezoidal integration weights, satisfying \(a_l\geq0\) and \(\sum_l a_l=1\). I summarize the grid alpha path with four functionals:
	\begin{align}
		SA_m
		&=
		\sum_{l=1}^{L}a_l\alpha_m(p_l),
		&
		IAE_m
		&=
		\sum_{l=1}^{L}a_l
		\left|
		\alpha_m(p_l)
		\right|,
		\label{eq:area_abs_continuous}\\
		ISE_m
		&=
		\sum_{l=1}^{L}a_l
		\alpha_m(p_l)^2,
		&
		SUP_m
		&=
		\max_{1\leq l\leq L}
		\left|
		\alpha_m(p_l)
		\right|.
		\label{eq:square_sup_continuous}
	\end{align}
	Their sample counterparts replace \(\alpha_m(p_l)\) with \(\widehat{\alpha}_m(p_l)\).
	
	The signed area \(SA_m\) measures directional tilt: a positive value indicates that the prefix tends to outperform its equal-exposure market benchmark as the axis is traversed, and a negative value the reverse. The magnitude measures \(IAE_m\) and \(ISE_m\) do not permit positive and negative regions to cancel, and \(SUP_m\) records the largest local absolute pricing error on the grid.
	
	I also report the coherence ratio
	\begin{equation}
		CR_m
		=
		\frac{|SA_m|}{IAE_m},
		\qquad
		0\leq CR_m\leq1,
		\label{eq:coherence_ratio}
	\end{equation}
	when \(IAE_m>0\). A ratio near one indicates a predominantly one-signed bridge-alpha curve; a smaller ratio indicates that positive and negative portions offset in the signed integral.
	
	The grid is the maintained test-asset coordinate, so the reported functionals are finite-grid statistics rather than claims about an unobserved continuum. The continuous notation provides an economic interpretation of their limiting object, while the empirical test is conducted on the same \(L\)-point grid used to construct the bridge returns.
	
	\subsection{Linear closure of the zero-alpha restrictions}
	\label{subsec:restricted_sdf_sufficiency}
	
	The diagnostic has a restricted linear interpretation. Let \(R_M^e\) denote the excess return on the implemented aggregate market, and define
	\begin{equation}
		\mathcal{V}_{cap,\mathcal{P}}
		=
		\operatorname{span}
		\left\{
		R_M^e,
		D(p_1),\ldots,D(p_L)
		\right\},
		\label{eq:cap_axis_span}
	\end{equation}
	the finite-dimensional return space generated by the aggregate market and the implemented cap-axis bridges.
	
	\begin{assumption}[Second moments]
		\label{asm:second_moments}
		All test returns and factor returns have finite second moments, and
		\(\operatorname{Var}(f_m)\) is nonsingular.
	\end{assumption}
	
	For any square-integrable return \(R\), the population OLS intercept relative to model \(m\) is
	\begin{equation}
		\alpha_m(R)
		=
		\mathbb{E}[R]
		-
		\mathbb{E}[f_m]'
		\operatorname{Var}(f_m)^{-1}
		\operatorname{Cov}(f_m,R).
		\label{eq:alpha_functional_explicit}
	\end{equation}
	Under \Cref{asm:second_moments}, \(\alpha_m(\cdot)\) is a linear functional of the test return.
	
	\begin{proposition}[Linear closure of cap-axis zero-alpha restrictions]
		\label{prop:cap_axis_sufficiency}
		Under \Cref{asm:second_moments}, model \(m\) has zero alpha on every return in \(\mathcal{V}_{cap,\mathcal{P}}\) if and only if
		\begin{equation}
			\alpha_m(R_M^e)=0
			\quad\text{and}\quad
			\alpha_m\{D(p_l)\}=0
			\quad
			\text{for }l=1,\ldots,L.
			\label{eq:cap_axis_sufficiency_condition}
		\end{equation}
	\end{proposition}
	
	\begin{proof}
		Necessity follows because each generator belongs to
		\(\mathcal{V}_{cap,\mathcal{P}}\). For sufficiency, every
		\(Z\in\mathcal{V}_{cap,\mathcal{P}}\) can be written as
		\(Z=b_0R_M^e+\sum_{l=1}^{L}b_lD(p_l)\)
		for fixed coefficients \(b_0,\ldots,b_L\), and linearity of
		\(\alpha_m(\cdot)\) gives
		\[
		\alpha_m(Z)
		=
		b_0\alpha_m(R_M^e)
		+
		\sum_{l=1}^{L}
		b_l\alpha_m\{D(p_l)\}
		=
		0.
		\]
	\end{proof}
	
	The proposition is an algebraic closure result, not a claim of full model validity; it concerns the return series generated by the implemented bridges and their fixed linear combinations.
	
	An additional distinction arises between the fractional benchmark and the whole-stock implementation. Under fractional assignment,
	\begin{equation}
		C_t^{e,\ast}(p)
		=
		D_t^\ast(p)+pR_{M,t}^e,
		\label{eq:fractional_prefix_excess_identity}
	\end{equation}
	where \(C_t^{e,\ast}(p)=C_t^\ast(p)-pR_{f,t}\). Because \(p\) is constant, zero alpha on the market and all fractional bridges implies zero alpha on fractional prefixes, tails, intervals, and their fixed step-function combinations. For the implemented bridge, however,
	\begin{equation}
		C_t^e(p)
		=
		D_t(p)+s_t(p)R_{M,t}^e,
		\label{eq:implemented_prefix_excess_identity}
	\end{equation}
	and \(s_t(p)\) varies over time as portfolio weights evolve and whole stocks cross target cutoffs. Therefore, zero alpha on \(D(p)\) and \(R_M^e\) does not mechanically imply zero alpha on the unscaled prefix return \(C^e(p)\). The empirical claim is correspondingly narrower and more precise: the model is evaluated on equal-exposure prefix-versus-market bridges, not on raw prefix portfolios considered in isolation.
	
	\subsection{Fixed-grid null distribution}
	\label{subsec:functional_null_distribution}
	
	Stack the \(L\) bridge returns into \(\boldsymbol{D}_t=(D_t(p_1),\ldots,D_t(p_L))'\), let \(x_{m,t}=(1,f_{m,t}')'\), and let \(\widehat{\boldsymbol{\alpha}}_m\) be the intercept row from the multivariate OLS regression of \(\boldsymbol{D}_t\) on \(x_{m,t}\). Under standard moment, stationarity, and weak-dependence conditions,
	\begin{equation}
		\sqrt{T}
		\left(
		\widehat{\boldsymbol{\alpha}}_m
		-
		\boldsymbol{\alpha}_m
		\right)
		\Rightarrow
		N(0,\Omega_m),
		\label{eq:alpha_vector_asymptotic}
	\end{equation}
	where \(\Omega_m\) is the long-run covariance matrix of the grid-level intercept influence process. Under the zero-curve null, the finite-grid sampling approximation is therefore
	\begin{equation}
		\widetilde{\boldsymbol{\alpha}}^{(b)}_{m,0}
		\sim
		N
		\left(
		\boldsymbol{0},
		\frac{\widehat{\Omega}_m}{T}
		\right).
		\label{eq:finite_grid_gaussian_null}
	\end{equation}
	For each Gaussian draw, I apply the same trapezoidal weights and grid maximum used for the observed curve:
	\begin{align}
		\widetilde{IAE}^{(b)}_{m,0}
		&=
		\sum_{l=1}^{L}
		a_l
		\left|
		\widetilde{\alpha}^{(b)}_{m,0}(p_l)
		\right|,
		\\
		\widetilde{ISE}^{(b)}_{m,0}
		&=
		\sum_{l=1}^{L}
		a_l
		\widetilde{\alpha}^{(b)}_{m,0}(p_l)^2,
		\\
		\widetilde{SUP}^{(b)}_{m,0}
		&=
		\max_l
		\left|
		\widetilde{\alpha}^{(b)}_{m,0}(p_l)
		\right|.
		\label{eq:simulated_functional_nulls}
	\end{align}
	The upper-tail probability of the observed functional relative to these draws gives its HAC-Gaussian-process \(p\)-value. Calling this a Gaussian-process procedure emphasizes the joint covariance structure across the grid; computationally, it is a multivariate Gaussian approximation on the fixed empirical grid.
	
	The directional statistic is treated separately. I test the alpha of the rank-area return \(Y_t\) using a conventional HAC standard error, supplying a scalar directional test without requiring a nonlinear transformation of the full alpha vector. The grid-integrated \(\widehat{SA}_m\) remains the reported curve statistic, while the rank-area alpha and their difference provide a separate implementation and boundary-quantization check.
	
	\subsection{Interpretation of the diagnostic}
	\label{subsec:interpretation_of_pass}
	
	The population null \(\boldsymbol{\alpha}_m=\boldsymbol{0}\) means that model \(m\) leaves no alpha on any implemented equal-exposure bridge along the tested cap-axis grid. Conditional on the aggregate-market consistency gate, it also implies zero alpha on every fixed linear combination of these bridge returns by \Cref{prop:cap_axis_sufficiency}.
	
	This restriction is more informative than testing a single body--tail cutoff. Each \(p_l\) asks whether the capitalization prefix ending at that location differs systematically from an equal-exposure position in the whole market. The full curve records where such deviations arise, while \(SA\), \(IAE\), \(ISE\), and \(SUP\) separate direction, aggregate magnitude, quadratic concentration, and local extremity.
	
	At the same time, the interpretation should not be overstated. Because the implemented exposure \(s_t(p_l)\) is time varying, passing the bridge test is not algebraically equivalent to pricing every raw prefix or size portfolio. It is also not full SDF validity and does not imply that the model dominates in factor spanning or maximum-Sharpe terms. It establishes zero-alpha consistency only on the tested bridge-return space.
	
	A statistical nonrejection is weaker than the population null: it means that the estimated directional and magnitude functionals are not large relative to their sampling distributions at the specified formation and observation frequencies. Conversely, rejection identifies a systematic pricing-error footprint on the cap-axis bridge family, and the shape of the estimated curve localizes the capitalization region in which the model-relative error accumulates.
	
	The same object can be used at two levels. At the model level, the bridge curve tests the cap-axis zero-alpha consistency of an assembled factor model. At the factor level, its magnitude functionals provide a coordinate describing the footprint left when an individual factor is added to the market. Neither use replaces mean--variance spanning; the cap-axis diagnostic records a distinct restriction on how factor-model pricing errors are distributed inside the aggregate market.
	
	\section{Data and Methodology}
	\label{sec:data}
	
	This section describes the empirical implementation of the cap-axis zero-alpha diagnostic. I construct an investable CRSP market, verify that its aggregate return is close to standard provider market series, and implement the whole-stock bridge under annual Jul--Jun formation and buy-and-hold accounting. Formal inference for the nonlinear bridge functionals uses a finite-grid HAC-Gaussian-process null. A separate residual-block calibration evaluates finite-sample size and power. No stock-ordering randomization is used.
	
	\subsection{Constructing the CRSP investable universe}
	\label{subsec:crsp-universe}
	
	I use CRSP individual-stock data \citep{CRSP}, the Fama--French Data Library \citep{FrenchDataLibrary}, and the q-factor provider data \citep{globalqFactors}. The stock-level sample runs from January 3, 1967 to December 31, 2024. The base universe consists of common stocks listed on the NYSE, AMEX, and NASDAQ. The investability filter removes the extreme microcap and illiquidity tail while retaining nearly all of the market's economic weight and trading activity.
	
	Investability is updated at each month-end using market capitalization and liquidity. The month-end market capitalization of stock \(i\) is
	\begin{equation}
		ME_{i,t}
		=
		\left|PRC_{i,t}\right|
		\times SHROUT_{i,t}
		\times 1{,}000.
		\label{eq:me_definition_data}
	\end{equation}
	After sorting stocks in descending \(ME\), define the cumulative market-cap share of stock \(i\) as
	\begin{equation}
		CumME_{i,t}
		=
		\frac{
			\sum_{j:ME_{j,t}\geq ME_{i,t}}ME_{j,t}
		}{
			\sum_jME_{j,t}
		}.
		\label{eq:cumme_definition_data}
	\end{equation}
	
	Liquidity is measured by average dollar volume over the most recent 63 trading days:
	\begin{align}
		DVOL_{i,d}
		&=
		\left|PRC_{i,d}\right|VOL_{i,d},\\
		ADV63_{i,t}
		&=
		\frac{1}{N_{i,t}}
		\sum_{d\in\mathcal{W}_t}DVOL_{i,d},
		\label{eq:adv63_definition_data}
	\end{align}
	where \(\mathcal{W}_t\) is the 63-trading-day window ending at month-end \(t\), and \(N_{i,t}\) is the number of days with observed volume in that window. I use monthly cross-sectional liquidity percentiles rather than a fixed dollar-volume threshold so that the screen adapts to changes in the nominal scale of the market over the long sample.
	
	The selection rule uses hysteresis. An incumbent exits if \(CumME_{i,t}>0.999\) or if its \(ADV63\) falls in the bottom 2.5\% of that month's cross-section. A non-constituent enters only if \(CumME_{i,t}\leq0.995\) and its \(ADV63\) is at least at the fifth percentile. The stricter entry condition reduces mechanical turnover near the boundary. For early NASDAQ observations, I defer the liquidity screen until 63 trading days of volume history are available while continuing to apply the market-cap screen. All screening variables are observed by month-end, and the resulting universe is first applied on the following month's first trading day.
	
	The final universe contains, on average, 77.6\% of the base common-stock sample while preserving approximately 99.7\% of total market capitalization and 63-day average dollar volume. At the end of 2024, it contains 2,426 of the 3,804 base-sample stocks, with market-cap and dollar-volume preservation rates of 99.78\% and 99.28\%. The screen therefore removes a large number of economically negligible and illiquid securities without materially changing the scale of the aggregate market.
	
	\subsection{Aggregate-market construction and validation}
	\label{subsec:market-replication}
	
	I first construct a daily value-weighted market return, denoted \(cMKT\), from the investable CRSP universe. This production-market series is used to verify the return accounting, delisting treatment, and coverage of the stock panel against standard market-factor providers. The month-end universe is applied from the next month's first trading day, returns include dividend-inclusive and observable delisting returns, and portfolio cash flows are reinvested proportionally across surviving positive-value positions.
	
	This production-market validation should be distinguished from the mode-specific market benchmark inside each bridge. Under annual Jul--Jun formation, the bridge is constructed against the buy-and-hold aggregate portfolio generated by the same annual formation schedule; under alternative formation cycles, it is constructed against the corresponding mode-specific aggregate portfolio. Endpoint closure and body--tail recombination therefore hold relative to the market portfolio actually used by that formation design. The \(cMKT\) comparison is a data and accounting validation rather than a formal premise of the zero-investment bridge test.
	
	I compare \(cMKT\) with the total-market returns implied by the Fama--French and q5 provider series. In the common sample, the correlation between \(cMKT\) and the Fama--French market return is 0.999778, with \(R^2=0.999556\), daily RMSE of 2.21 bp, and mean absolute error of 1.35 bp. Against the q5 market return, the correlation is 0.999799, with \(R^2=0.999598\), RMSE of 2.10 bp, and mean absolute error of 1.27 bp. The q5 factors used below are the provider's distributed daily and monthly series, not frequency conversions constructed from one another.
	
	\begin{figure}[!ht]
		\centering
		\includegraphics[width=0.95\linewidth]{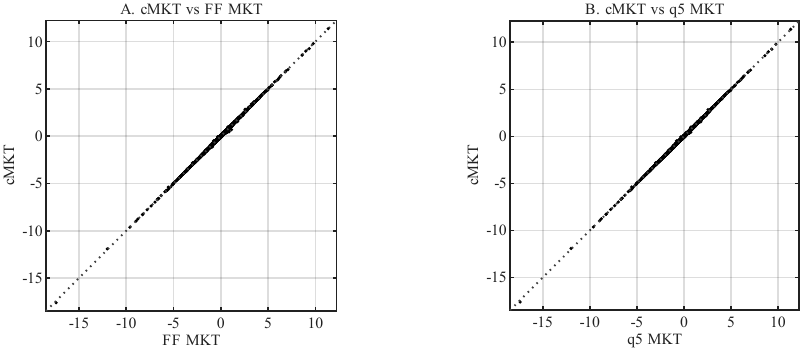}
		\caption{CRSP-based market return versus standard market factors}
		\label{fig:mkt-replication}
		\begin{minipage}{0.95\linewidth}
			\footnotesize
			\emph{Note:} The figure compares the market return \(cMKT\), constructed from the CRSP investable universe, with the Fama--French and q5 market returns. The comparison is on a total-market-return basis, with each provider's risk-free rate added back to its excess market return. The dashed line is the 45-degree line. The common sample runs from January 3, 1967 to December 31, 2024.
		\end{minipage}
	\end{figure}
	
	As an additional consistency check, I regress the whole investable-market excess return on each candidate factor model. Since every candidate contains a market factor, a material intercept would indicate a mismatch in the aggregate-market implementation before the internal cap-axis decomposition is considered. \Cref{tab:market-sanity-check} reports the results. Every regression has \(R^2\) above 0.999, and every alpha \(p\)-value exceeds 0.05 under Newey--West standard errors with a 21-trading-day lag.
	
	\begin{table}[!ht]
		\centering
		\singlespacing
		\caption{Market sanity check on the whole market return}
		\label{tab:market-sanity-check}
		\footnotesize
		\setlength{\tabcolsep}{4.8pt}
		\begin{tabular}{lrrrrr}
			\toprule
			Model & Observations & Annual alpha (bp) & \(t(\alpha)\) & \(p(\alpha)\) & \(R^2\) \\
			\midrule
			CAPM    & 14,598 & -1.02 & -0.22 & 0.829 & 0.999556 \\
			FF3     & 14,598 &  0.21 &  0.05 & 0.963 & 0.999571 \\
			Carhart & 14,598 & -0.06 & -0.01 & 0.989 & 0.999571 \\
			FF5     & 14,598 & -1.07 & -0.24 & 0.812 & 0.999572 \\
			FF6     & 14,598 & -1.18 & -0.26 & 0.792 & 0.999572 \\
			q5      & 14,598 & -5.78 & -1.44 & 0.151 & 0.999619 \\
			\bottomrule
		\end{tabular}
		
		\vspace{0.4em}
		\parbox{0.95\textwidth}{\footnotesize
			\emph{Note:} The table reports regressions of the whole investable-market excess return on each candidate factor model. Alpha is annualized by multiplying the daily intercept by 252 and is reported in basis points. The sample runs from January 3, 1967 to December 31, 2024. \(p\)-values use Newey--West standard errors with a 21-trading-day lag.}
	\end{table}
	
	Passing this check does not imply that every portfolio constructed from the universe has zero alpha. It shows that the stock panel, return accounting, and provider market factors agree closely at the aggregate level. The cap-axis analysis then asks whether zero-investment equal-exposure bridges reveal model-relative pricing errors within that market.
	
	\subsection{Annual formation and buy-and-hold accounting}
	\label{subsec:design-frequency}
	
	The baseline design uses annual Jul--Jun formation. On the first trading day of each July, eligible stocks are sorted in descending order of prior market capitalization, and initial holdings are proportional to that market capitalization. No new stocks are admitted until the next annual formation date, and the formation-date capitalization order is held fixed throughout the holding year.
	
	Between formation dates, portfolio positions evolve on a buy-and-hold basis. Let \(H_{i,t}\) denote stock \(i\)'s beginning-of-day holding value, with start-of-day weight \(w_{i,t}=H_{i,t}/\sum_jH_{j,t}\). Bridge returns, rank-area returns, and the mode-specific market return are computed from these start-of-day weights and dividend-inclusive total returns.
	
	To update holdings, I separate the ex-dividend return \(RETX_{i,t}\) from the dividend-inclusive return \(RETADJ_{i,t}\). The ex-dividend position value is \(H^{x}_{i,t+1}=H_{i,t}(1+RETX_{i,t})\), and portfolio-level dividend cash is
	\begin{equation}
		DivCash_t
		=
		\sum_i
		H_{i,t}
		\left(
		RETADJ_{i,t}-RETX_{i,t}
		\right).
		\label{eq:dividend_cash_pool}
	\end{equation}
	This cash is reinvested in proportion to positive ex-dividend position values:
	\begin{equation}
		H_{i,t+1}
		=
		H^{x}_{i,t+1}
		+
		DivCash_t
		\frac{
			H^{x}_{i,t+1}
		}{
			\sum_jH^{x}_{j,t+1}
		}.
		\label{eq:dividend_pool_reinvestment}
	\end{equation}
	The calculation is rescaled when necessary to preserve the portfolio's total post-return value. Observable delisting returns are incorporated into \(RETADJ\). This accounting lets portfolio weights drift naturally while avoiding an artificial assumption that each stock's cash distribution is reinvested only into that same stock.
	
	The annual formation schedule determines the ordering and holding process; the observation frequency determines how the resulting payoff is evaluated. At the daily frequency, daily bridge returns are regressed on provider daily factors. At the monthly frequency, the same daily bridge path is aggregated to calendar months and regressed on the provider's distributed monthly factors.
	
	For an aggregate long-only return, calendar-month aggregation uses compounding:
	\begin{equation}
		R_{M,m}
		=
		\prod_{d\in m}
		\left(
		1+R_{M,d}
		\right)-1.
		\label{eq:monthly_market_compounding}
	\end{equation}
	The bridge and rank-area returns are zero-investment return contributions and are summed within the month:
	\begin{align}
		D_m(p_l)
		&=
		\sum_{d\in m}D_d(p_l),\\
		Y_m
		&=
		\sum_{d\in m}Y_d.
		\label{eq:monthly_zero_investment_sum}
	\end{align}
	The monthly analysis is therefore not obtained by averaging daily coefficient estimates. It reconstructs monthly observations from the same annual-formation bridge payoff and evaluates them using provider monthly factors.
	
	The daily--monthly comparison is reduced-form. It captures both the temporal aggregation of the test return and differences between providers' daily and monthly factor implementations. I interpret it as a horizon diagnostic rather than as a clean decomposition of these two channels.
	
	\subsection{Whole-stock bridge implementation}
	\label{subsec:empirical-integral-engine}
	
	The empirical engine implements the finite-market object defined in \Cref{subsec:discrete_integral_definition}. The baseline grid is
	\begin{equation}
		\mathcal{P}
		=
		\{0,0.005,\ldots,0.995,1\},
		\qquad
		L=201.
		\label{eq:baseline_cap_grid}
	\end{equation}
	At each return date, the stocks retain their formation-date order, while their start-of-day buy-and-hold weights determine the cumulative market exposure along that order.
	
	For grid target \(p_l\), define
	\begin{equation}
		k(l,t)
		=
		\begin{cases}
			0, & p_l=0,\\[2pt]
			\min\left\{
			k:
			\sum_{i=1}^{k}w_{i,t}\geq p_l
			\right\},
			& p_l\in(0,1],
		\end{cases}
		\qquad
		s_{l,t}
		=
		\sum_{i=1}^{k(l,t)}w_{i,t}.
		\label{eq:empirical_realized_share_concise}
	\end{equation}
	The engine includes the boundary stock in full. It records \(s_{l,t}\) separately so that the difference between the target coordinate \(p_l\) and the realized whole-stock exposure is observable.
	
	The prefix contribution and aggregate-market return are
	\begin{align}
		B_{l,t}
		&=
		\sum_{i=1}^{k(l,t)}
		w_{i,t}R_{i,t},\\
		R^M_t
		&=
		\sum_{i=1}^{N_t}
		w_{i,t}R_{i,t}.
	\end{align}
	The implemented bridge is
	\begin{equation}
		D_{l,t}
		=
		B_{l,t}
		-
		s_{l,t}R^M_t.
		\label{eq:empirical_bridge_return_concise}
	\end{equation}
	The prefix and market legs have the same realized investment exposure \(s_{l,t}\), so \(D_{l,t}\) is zero-investment, and the endpoint observations are set by construction to \(D_{1,t}=D_{L,t}=0\).
	
	For each date and grid point, I also construct the complementary tail bridge and verify \(D_{l,t}+D^T_{l,t}=0\). The maximum recombination residual is compared with a tolerance of \(10^{-8}\); in the completed runs, the identity holds to machine precision. I additionally verify that the prefix and tail return contributions recombine to the mode-specific aggregate-market return.
	
	\subsection{Factor models and common samples}
	\label{subsec:factor-models-samples}
	
	I compare six model specifications. CAPM contains the market excess return. FF3 contains the market, SMB, and HML. Carhart adds momentum to FF3. FF5 contains the market, SMB, HML, RMW, and CMA. FF6 adds momentum to FF5. The q5 model contains its provider market factor together with the ME, IA, ROE, and expected-growth factors.
	
	For model \(m\) and grid point \(p_l\), I estimate
	\begin{equation}
		D_t(p_l)
		=
		\alpha_m(p_l)
		+
		\beta_m(p_l)'f_{m,t}
		+
		\varepsilon_{m,t}(p_l).
		\label{eq:empirical_bridge_alpha_regression}
	\end{equation}
	Because the bridge is zero-investment, it enters the regression without subtracting the risk-free rate. The provider factors are expressed in return units consistent with the bridge panel before estimation.
	
	All model comparisons within a frequency use the common intersection of the bridge dates and the six factor-provider samples. The daily common sample contains 14,598 trading days from January 3, 1967 to December 31, 2024. The monthly common sample contains 696 calendar months. This common-sample rule prevents differences in model statistics from being driven by different estimation windows.
	
	The coefficient matrix is estimated jointly across the grid. Let
	\[
	X_m
	=
	\begin{bmatrix}
		1 & f_{m,1}'\\
		\vdots & \vdots\\
		1 & f_{m,T}'
	\end{bmatrix},
	\qquad
	D
	=
	\begin{bmatrix}
		D_1(p_1) & \cdots & D_1(p_L)\\
		\vdots & \ddots & \vdots\\
		D_T(p_1) & \cdots & D_T(p_L)
	\end{bmatrix}.
	\]
	Then
	\begin{equation}
		\widehat{B}_m
		=
		(X_m'X_m)^{-1}X_m'D,
		\label{eq:bridge_multivariate_ols}
	\end{equation}
	and the first row of \(\widehat{B}_m\) is \(\widehat{\boldsymbol{\alpha}}_m'=(\widehat{\alpha}_m(p_1),\ldots,\widehat{\alpha}_m(p_L))\). This joint matrix expression is computationally equivalent to estimating the \(L\) bridge regressions separately with the same regressors.
	
	\subsection{Rank-area and grid functionals}
	\label{subsubsec:rank-area-scalar-diagnostic-concise}
	
	For each date, define stock \(i\)'s midpoint cumulative share as
	\begin{equation}
		m_{i,t}
		=
		\sum_{j<i}w_{j,t}
		+
		\frac{1}{2}w_{i,t}.
		\label{eq:midpoint_share_concise}
	\end{equation}
	The rank-area return is
	\begin{equation}
		Y_t
		=
		\sum_{i=1}^{N_t}
		w_{i,t}R_{i,t}
		\left(
		\frac{1}{2}-m_{i,t}
		\right).
		\label{eq:rank_area_return_concise}
	\end{equation}
	I regress \(Y_t\) on the same factor matrix used for the bridge curve; its intercept is the rank-area alpha. The daily regression uses a Newey--West lag of 21 trading days, and the monthly regression uses a lag of six months.
	
	The reported curve statistics are computed directly from the estimated grid alpha path. With trapezoidal weights \(a_l\),
	\begin{align}
		\widehat{SA}_m
		&=
		\sum_{l=1}^{L}
		a_l\widehat{\alpha}_m(p_l),\\
		\widehat{IAE}_m
		&=
		\sum_{l=1}^{L}
		a_l
		\left|
		\widehat{\alpha}_m(p_l)
		\right|,\\
		\widehat{ISE}_m
		&=
		\sum_{l=1}^{L}
		a_l
		\widehat{\alpha}_m(p_l)^2,\\
		\widehat{SUP}_m
		&=
		\max_l
		\left|
		\widehat{\alpha}_m(p_l)
		\right|.
		\label{eq:empirical_bridge_functionals}
	\end{align}
	
	The grid-integrated \(\widehat{SA}_m\) and the rank-area alpha are not imposed to be identical. The former integrates the whole-stock bridge curve over the fixed grid; the latter implements the fractional midpoint representation as one return. I report \(\widehat{SA}_m\) as the curve's signed area and use the rank-area alpha and HAC \(t\)-statistic as its scalar directional companion. Their annualized difference is recorded as a boundary-quantization audit and remains below 1.5 bp for every baseline model.
	
	The linear statistics \(\widehat{SA}_m\) and the rank-area alpha preserve sign. The nonlinear magnitude functionals \(\widehat{IAE}_m\), \(\widehat{ISE}_m\), and \(\widehat{SUP}_m\) require joint inference for the full alpha vector, described next.
	
	\subsection{Finite-grid HAC-Gaussian-process inference}
	\label{subsubsec:hac_gp_inference}
	
	Formal inference for the nonlinear functionals is based on the joint long-run covariance matrix of the estimated intercept vector. Let \(x_{m,t}\) denote the \(t\)-th row of \(X_m\), written as a column, and let \(e_1\) select the intercept. Define
	\begin{equation}
		\widehat{Q}_m
		=
		\frac{1}{T}X_m'X_m.
		\label{eq:factor_second_moment_matrix}
	\end{equation}
	Let \(\widehat{\boldsymbol{\varepsilon}}_{m,t}=(\widehat{\varepsilon}_{m,t}(p_1),\ldots,\widehat{\varepsilon}_{m,t}(p_L))'\) be the residual vector across the cap-axis grid. The estimated intercept influence vector is
	\begin{equation}
		\widehat{\boldsymbol{\psi}}_{m,t}
		=
		\left(
		e_1'
		\widehat{Q}_m^{-1}
		x_{m,t}
		\right)
		\widehat{\boldsymbol{\varepsilon}}_{m,t}.
		\label{eq:alpha_influence_vector}
	\end{equation}
	
	I estimate the \(L\times L\) long-run covariance matrix of this process with a multivariate Newey--West estimator:
	\begin{equation}
		\widehat{\Omega}_m
		=
		\widehat{\Gamma}_{m,0}
		+
		\sum_{h=1}^{H}
		\left(
		1-\frac{h}{H+1}
		\right)
		\left(
		\widehat{\Gamma}_{m,h}
		+
		\widehat{\Gamma}_{m,h}'
		\right),
		\label{eq:cross_grid_hac}
	\end{equation}
	where
	\begin{equation}
		\widehat{\Gamma}_{m,h}
		=
		\frac{1}{T}
		\sum_{t=h+1}^{T}
		\widehat{\boldsymbol{\psi}}_{m,t}
		\widehat{\boldsymbol{\psi}}_{m,t-h}'.
		\label{eq:cross_grid_hac_lag}
	\end{equation}
	The baseline daily lag is \(H=21\) trading days, and the baseline monthly lag is \(H=6\) months.
	
	Under the zero-curve null, I simulate
	\begin{equation}
		\widetilde{\boldsymbol{\alpha}}_{m,0}^{(b)}
		\sim
		N
		\left(
		\boldsymbol{0},
		\frac{\widehat{\Omega}_m}{T}
		\right),
		\qquad
		b=1,\ldots,B.
		\label{eq:empirical_hac_gp_draw}
	\end{equation}
	Each draw is expressed in the same return units as the estimated alpha vector, so the observed and simulated functionals can be compared directly without additional asymptotic rescaling.
	
	For each draw, I compute
	\begin{align}
		\widetilde{IAE}_{m,0}^{(b)}
		&=
		\sum_{l=1}^{L}
		a_l
		\left|
		\widetilde{\alpha}_{m,0}^{(b)}(p_l)
		\right|,\\
		\widetilde{ISE}_{m,0}^{(b)}
		&=
		\sum_{l=1}^{L}
		a_l
		\widetilde{\alpha}_{m,0}^{(b)}(p_l)^2,\\
		\widetilde{SUP}_{m,0}^{(b)}
		&=
		\max_l
		\left|
		\widetilde{\alpha}_{m,0}^{(b)}(p_l)
		\right|.
		\label{eq:empirical_hac_gp_functionals}
	\end{align}
	The upper-tail \(p\)-value for functional \(G\in\{IAE,ISE,SUP\}\) is
	\begin{equation}
		p^{HAC\text{-}GP}_{G,m}
		=
		\frac{
			1+
			\sum_{b=1}^{B}
			\mathbf{1}
			\left\{
			\widetilde{G}_{m,0}^{(b)}
			\geq
			\widehat{G}_m
			\right\}
		}{
			B+1
		}.
		\label{eq:hac_gp_empirical_pvalue}
	\end{equation}
	The baseline uses \(B=50{,}000\) antithetic Gaussian draws.
	
	Before simulation, I symmetrize \(\widehat{\Omega}_m\). A dense grid can produce small negative sample eigenvalues from numerical error or near-collinearity. I therefore impose a small eigenvalue floor, record the size of the adjustment, and simulate from the resulting positive-semidefinite matrix. This correction affects numerical factorization rather than the observed alpha path.
	
	The rank-area alpha is tested separately with its Newey--West standard error. I do not form conventional \(t\)-statistics for \(IAE\), \(ISE\), or \(SUP\), because these statistics are nonlinear, nonnegative functions of a highly dependent alpha vector.
	
	\subsection{Finite-sample size and power calibration}
	\label{subsubsec:finite-sample-calibration}
	
	I evaluate the finite-sample behavior of the nonlinear functional tests using the annual Jul--Jun bridge over the 1967--2024 monthly sample. The calibration holds each model's factor matrix fixed and uses its estimated residual curve as the sampling environment.
	
	For model \(m\), write the fitted bridge panel as
	\begin{equation}
		D
		=
		\boldsymbol{1}\widehat{\boldsymbol{\alpha}}_m'
		+
		F_m\widehat{B}_{m,f}
		+
		\widehat{E}_m,
		\label{eq:power_bridge_decomposition}
	\end{equation}
	where \(F_m\widehat{B}_{m,f}\) is the fitted non-intercept component and \(\widehat{E}_m\) is the \(T\times L\) residual panel. I center the residuals separately at every grid point.
	
	A simulation draw resamples the entire residual curve for each selected month rather than resampling grid points independently. The monthly indices are generated in six-month circular blocks. This preserves contemporaneous dependence across all 201 grid points and local serial dependence within each resampled block.
	
	The simulated bridge panel under alternative strength \(\lambda\) is
	\begin{equation}
		D_m^{(b)}(\lambda)
		=
		\boldsymbol{1}
		\left(
		\lambda
		\widehat{\boldsymbol{\alpha}}_m
		\right)'
		+
		F_m\widehat{B}_{m,f}
		+
		\widehat{E}_m^{(b)},
		\label{eq:power_simulated_bridge_panel}
	\end{equation}
	where \(\lambda\in\{0,0.25,0.50,0.75,1.00\}\). The case \(\lambda=0\) imposes the zero-curve null and measures empirical size. Positive \(\lambda\) values progressively inject the model-specific estimated cap-axis alpha curve and measure power against a structured alternative.
	
	For every simulated panel, I re-estimate the bridge regression and recompute \(\widehat{IAE}\), \(\widehat{ISE}\), and \(\widehat{SUP}\). Rejection is evaluated using two critical-value sources. The baseline source is the 5\% critical value from 50,000 HAC-Gaussian-process draws. The second is a residual block-bootstrap critical value based on 4,999 six-month circular-block draws. The block-bootstrap critical-value draws and the 4,999 evaluation draws are generated independently.
	
	The reported calibration averages rejection frequencies across the six models and the three nonlinear functionals. This aggregation summarizes whether the test family is approximately size-controlled under the null and whether rejection increases with alternative strength. Model-level power is also retained because the alternatives inherit the heterogeneous magnitudes of the estimated empirical curves.
	
	The residual-block calibration does not randomize the stock order, portfolio membership, market weights, or cap-rank coordinate. It evaluates the sampling behavior of the zero-alpha tests conditional on the observed annual-Jul bridge design. It therefore replaces, rather than supplements, an ordering-placebo exercise.
	
	\subsection{Auxiliary design checks}
	\label{subsec:auxiliary-design-checks}
	
	I examine two implementation channels separately from the finite-sample calibration. First, I vary the formation schedule. In addition to annual Jul--Jun formation, I construct daily, monthly, quarterly, and annual-Jan bridge paths. Daily formation re-sorts stocks on prior market capitalization every trading day and uses contemporaneous value-weighted return accounting. Non-daily schedules hold the formation order fixed and use the same buy-and-hold dividend-pool accounting as the annual-Jul baseline.
	
	Second, I apply a Dimson-style lead--lag correction to the daily bridge. For symmetric half-window \(k\), estimate
	\begin{equation}
		D_t(p_l)
		=
		\alpha_m^{(k)}(p_l)
		+
		\sum_{h=-k}^{k}
		\beta_{m,h}^{(k)}(p_l)'
		f_{m,t+h}
		+
		\varepsilon_{m,t}^{(k)}(p_l).
		\label{eq:dimson_bridge_regression}
	\end{equation}
	The corrected factor exposure is the sum of coefficients across the window,
	\begin{equation}
		\beta_{m,Dimson}^{(k)}(p_l)
		=
		\sum_{h=-k}^{k}
		\beta_{m,h}^{(k)}(p_l),
		\label{eq:dimson_summed_beta}
	\end{equation}
	and the corrected bridge alpha is the intercept \(\alpha_m^{(k)}(p_l)\). I use \(k\in\{1,2,3,5,8,10,13,21\}\). For every \(k\), the naive and corrected estimates are computed on the same \(k\)-trimmed sample so that attenuation is not driven by changing endpoints. The same correction is applied to the rank-area return.
	
	Alternative rebalancing schedules address whether the findings depend on one formation convention. The lead--lag exercise addresses whether daily bridge alphas are absorbed when factors are allowed to affect cap-rank portfolios with short delays. Neither exercise changes the formal HAC-GP null used in the main tables.
	
	\subsection{Factor-coordinate scan}
	\label{subsec:factor-scan-methodology}
	
	The model-level analysis is complemented by a monthly factor-coordinate scan. I use the same annual Jul--Jun bridge path and add each candidate factor \(g_j\) separately to the market:
	\begin{equation}
		D_t(p_l)
		=
		\alpha_j(p_l)
		+
		\beta_{j,M}(p_l)R^e_{M,t}
		+
		\beta_{j,g}(p_l)g_{j,t}
		+
		\varepsilon_{j,t}(p_l).
		\label{eq:factor_scan_bridge_regression}
	\end{equation}
	The factor's cap-axis footprint is summarized by the same \(SA\), \(IAE\), \(ISE\), and \(SUP\) functionals used in the model-level analysis.
	
	The candidate universe contains 155 monthly factors with sufficient overlap over the 1967--2024 common sample: 145 Global Factor Data candidates based on the replicated anomaly library of \citet{JKP23}, six Fama--French/Carhart factors, and four nonmarket q5 factors. The scan is descriptive and does not interpret 155 individual functional statistics as separate formal discoveries.
	
	For comparison with conventional factor performance, I calculate the annualized maximum-Sharpe-ratio gain from adding candidate \(g_j\) to the market. If \(\mu_j\) and \(\Sigma_j\) are the sample mean vector and covariance matrix of the relevant monthly factor set, its annualized maximum Sharpe ratio is
	\begin{equation}
		SR_j^{ann}
		=
		\sqrt{
			12\,
			\mu_j'
			\Sigma_j^{-1}
			\mu_j
		}.
		\label{eq:annualized_max_sharpe}
	\end{equation}
	The candidate's Sharpe contribution is
	\begin{equation}
		\Delta SR_j
		=
		SR_{\{M,g_j\}}^{ann}
		-
		SR_{\{M\}}^{ann}.
		\label{eq:factor_scan_delta_sr}
	\end{equation}
	
	I audit whether the cap-axis footprint merely reproduces conventional size-factor exposure using the canonical FF3 SMB factor. For each candidate other than FF3 SMB itself, estimate
	\begin{equation}
		g_{j,t}
		=
		a_j
		+
		b_{j,M}R^e_{M,t}
		+
		\gamma_{j,size}SMB^{FF3}_t
		+
		u_{j,t}.
		\label{eq:size_loading_audit_methodology}
	\end{equation}
	I then examine whether \(\left|\widehat{\gamma}_{j,size}\right|\) explains the cross-section of cap-axis magnitude norms and whether residualizing those norms with respect to the size loading changes their relation with \(\Delta SR_j\). The size loading is omitted for FF3 SMB itself because it is the benchmark size proxy.
	
	This scan uses the same bridge construction, common-sample alignment, and grid functionals as the main monthly analysis. Only the regression factor set changes from an assembled model to the market plus one candidate factor.
	
	\section{Empirical Results}
	\label{sec:result}
	
	This section reports the cap-axis diagnostic under the main annual Jul--Jun formation design. I evaluate the same cap-axis payoff at the daily and monthly factor frequencies. Robustness to alternative rebalancing cycles is reported in \Cref{sec:robust}.
	
	\subsection{Daily results}
	\label{subsec:daily-results}
	
	The daily sample runs from January 3, 1967 to December 31, 2024, with 14{,}598 trading days. Each model's bridge-alpha curve is estimated on a uniform grid of 201 cutoffs. \Cref{tab:daily-bridge} reports the grid functionals, the rank-area HAC statistic, Rank \(R^2\), the coherence ratio \(|\widehat{SA}|/\widehat{IAE}\), and formal zero-bridge-alpha \(p\)-values. The signed-area \(p\)-value is from the rank-area HAC test; the \(IAE\), \(ISE\), and \(SUP\) \(p\)-values are from the HAC-Gaussian-process functional null.
	
	\begin{table}[!htbp]
		\centering
		\singlespacing
		\caption{Annual Jul--Jun rebalancing, daily frequency: cap-axis bridge functional statistics and inference}
		\label{tab:daily-bridge}
		\footnotesize
		\setlength{\tabcolsep}{3pt}
		\begin{threeparttable}
			\begin{tabular}{lrrrrrrrrrrr}
				\toprule
				Model & $\widehat{SA}$ & $t_{NW}$ & Rank $R^2$ & $\widehat{IAE}$ & $\widehat{ISE}$ & $\widehat{SUP}$ & $\dfrac{|\widehat{SA}|}{\widehat{IAE}}$ & $p^{NW}_{SA}$ & $p^{HAC\text{-}GP}_{IAE}$ & $p^{HAC\text{-}GP}_{ISE}$ & $p^{HAC\text{-}GP}_{SUP}$ \\
				\midrule
				CAPM    & $-10.9$ & $-0.57$ & $0.003$ & $11.0$ & $186$    & $26.7$ & $0.99$ & $0.566$  & $0.616$  & $0.583$  & $0.566$  \\
				FF3     & $21.2$  & $1.94$  & $0.614$ & $21.7$ & $608$    & $42.8$ & $0.97$ & $0.053$  & $0.048$  & $0.055$  & $0.064$  \\
				Carhart & $20.8$  & $1.88$  & $0.614$ & $21.4$ & $587$    & $42.0$ & $0.97$ & $0.060$  & $0.057$  & $0.066$  & $0.077$  \\
				FF5     & $7.6$   & $0.71$  & $0.654$ & $10.3$ & $168$    & $29.0$ & $0.74$ & $0.479$  & $0.370$  & $0.356$  & $0.303$  \\
				FF6     & $8.2$   & $0.77$  & $0.654$ & $10.5$ & $175$    & $28.3$ & $0.78$ & $0.442$  & $0.377$  & $0.360$  & $0.337$  \\
				q5      & $-39.6$ & $-3.49$ & $0.650$ & $39.6$ & $1{,}910$ & $68.0$ & $1.00$ & $<0.001$ & $<0.001$ & $<0.001$ & $0.002$  \\
				\bottomrule
			\end{tabular}
			\begin{tablenotes}[flushleft]
				\footnotesize
				\item \emph{Note:} $\widehat{SA}$, $\widehat{IAE}$, and $\widehat{SUP}$ are in annualized basis points; $\widehat{ISE}$ is in $\text{bp}^2$. $\widehat{SA}$ is the grid-integrated signed area. $t_{NW}$, $p^{NW}_{SA}$, and Rank $R^2$ are from the rank-area portfolio regression, with Newey--West 21-trading-day standard errors. $|\widehat{SA}|/\widehat{IAE}$ is the coherence ratio measuring directionality conditional on nontrivial magnitude; values near one indicate a mostly one-signed curve. $p^{HAC\text{-}GP}_{IAE}$, $p^{HAC\text{-}GP}_{ISE}$, and $p^{HAC\text{-}GP}_{SUP}$ are upper-tail probabilities from the HAC-Gaussian-process functional null with 50{,}000 simulation draws. All results evaluate the annual Jul--Jun rebalancing portfolio on daily factors.
			\end{tablenotes}
		\end{threeparttable}
	\end{table}
	
	\begin{figure}[!htbp]
		\centering
		\includegraphics[width=3.2in]{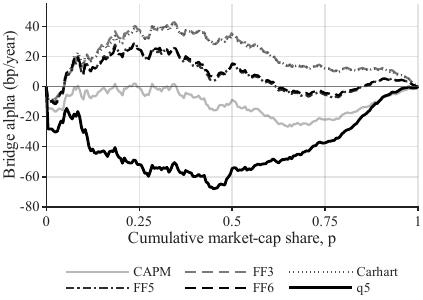}
		\caption{Annual Jul--Jun rebalancing, daily frequency: cap-axis bridge-alpha curve}
		\label{fig:daily-bridge}
		\begin{minipage}{\textwidth}
			\footnotesize
			\emph{Note:} Each candidate model's estimated bridge-alpha curve $\widehat{\alpha}_m(p)$ is plotted against the cumulative market-cap share $p\in[0,1]$. The vertical axis is in annualized basis points. The q5 curve reaches a minimum of about $-68$\,bp near $p\approx0.45$, stays in the negative region across the large- to mid-cap range, and closes to zero only at $p=1$, never changing sign along the cap-axis ($|\widehat{SA}|/\widehat{IAE}=1.00$). FF3 and Carhart instead form a positive one-signed curve, and FF5 and FF6 change sign near zero. The sample and estimation procedure are as in \Cref{tab:daily-bridge}.
		\end{minipage}
	\end{figure}
	
	\paragraph{Daily q5 distortion.}
	q5 is the only model with a strong daily cap-axis distortion under the formal zero-bridge-alpha inference. Its signed area is \(-39.6\) bp with a rank-area HAC statistic of \(-3.49\) and \(p^{NW}_{SA}<0.001\). The nonlinear functionals give the same conclusion: \(p^{HAC\text{-}GP}_{IAE}<0.001\), \(p^{HAC\text{-}GP}_{ISE}<0.001\), and \(p^{HAC\text{-}GP}_{SUP}=0.002\). The curve is also almost perfectly directional: \(|\widehat{SA}|/\widehat{IAE}=1.00\), with a local minimum near \(-68\) bp.
	
	\paragraph{Other daily models.}
	FF5, FF6, and CAPM do not reject cap-axis internal consistency in the main daily design. FF3 and Carhart form a positive bridge with the opposite sign to q5, but their status is weaker. FF3 is borderline, with \(p^{NW}_{SA}=0.053\) and \(p^{HAC\text{-}GP}_{IAE}=0.048\), while its \(ISE\) and \(SUP\) functionals are just above conventional levels. Carhart is similar but slightly weaker. I therefore treat these positive daily bridges as suggestive rather than central.
	
	\paragraph{Interpretation.}
	The daily result is axis-restricted. It does not imply that q5 fails to price the aggregate market, which it passes in \Cref{tab:market-sanity-check}, nor does it overturn mean--variance comparisons. It shows that, conditional on the aggregate-market gate, q5 leaves a one-sided pricing-error path along the cap-rank axis at the daily factor frequency. The next subsection asks whether this path survives at the monthly frequency.
	
	\FloatBarrier
	
	\subsection{Monthly results}
	\label{subsec:monthly-results}
	
	I next aggregate the same daily bridge payoffs to calendar months and regress them on the providers' distributed monthly factors. The sample has 696 months. The signed area is tested through the rank-area portfolio with a Newey--West 6-month lag, and the \(IAE\), \(ISE\), and \(SUP\) \(p\)-values are computed from the HAC-Gaussian-process functional null. \Cref{tab:monthly-bridge} reports the monthly functionals and inference statistics.
	
	\begin{table}[!htbp]
		\centering
		\singlespacing
		\caption{Annual Jul--Jun rebalancing, monthly frequency: cap-axis bridge functional statistics and inference}
		\label{tab:monthly-bridge}
		\footnotesize
		\setlength{\tabcolsep}{3pt}
		\begin{threeparttable}
			\begin{tabular}{lrrrrrrrrrrr}
				\toprule
				Model & $\widehat{SA}$ & $t_{NW}$ & Rank $R^2$ & $\widehat{IAE}$ & $\widehat{ISE}$ & $\widehat{SUP}$ & $\dfrac{|\widehat{SA}|}{\widehat{IAE}}$ & $p^{NW}_{SA}$ & $p^{HAC\text{-}GP}_{IAE}$ & $p^{HAC\text{-}GP}_{ISE}$ & $p^{HAC\text{-}GP}_{SUP}$ \\
				\midrule
				CAPM    & $11.6$ & $0.59$ & $0.091$ & $12.3$ & $213$    & $30.0$ & $0.94$ & $0.555$ & $0.582$ & $0.569$ & $0.505$ \\
				FF3     & $32.4$ & $2.81$ & $0.649$ & $32.5$ & $1{,}293$ & $57.1$ & $1.00$ & $0.005$ & $0.005$ & $0.008$ & $0.013$ \\
				Carhart & $32.8$ & $2.60$ & $0.649$ & $33.0$ & $1{,}315$ & $57.4$ & $1.00$ & $0.009$ & $0.010$ & $0.014$ & $0.023$ \\
				FF5     & $32.5$ & $2.87$ & $0.680$ & $32.5$ & $1{,}293$ & $55.3$ & $1.00$ & $0.004$ & $0.004$ & $0.006$ & $0.017$ \\
				FF6     & $32.4$ & $2.72$ & $0.680$ & $32.4$ & $1{,}281$ & $56.4$ & $1.00$ & $0.007$ & $0.007$ & $0.011$ & $0.026$ \\
				q5      & $4.4$  & $0.40$ & $0.677$ & $5.2$  & $37$     & $13.6$ & $0.83$ & $0.690$ & $0.870$ & $0.910$ & $0.958$ \\
				\bottomrule
			\end{tabular}
			\begin{tablenotes}[flushleft]
				\footnotesize
				\item \emph{Note:} Units and definitions are as in \Cref{tab:daily-bridge}. $t_{NW}$, $p^{NW}_{SA}$, and Rank $R^2$ are from the monthly rank-area portfolio regression, with Newey--West 6-month standard errors. $p^{HAC\text{-}GP}_{IAE}$, $p^{HAC\text{-}GP}_{ISE}$, and $p^{HAC\text{-}GP}_{SUP}$ are upper-tail probabilities from the HAC-Gaussian-process functional null with 50{,}000 simulation draws. All results evaluate the annual Jul--Jun rebalancing portfolio after aggregation to monthly returns and assessment on provider monthly factors.
			\end{tablenotes}
		\end{threeparttable}
	\end{table}
	
	\begin{figure}[!htbp]
		\centering
		\includegraphics[width=3.2in]{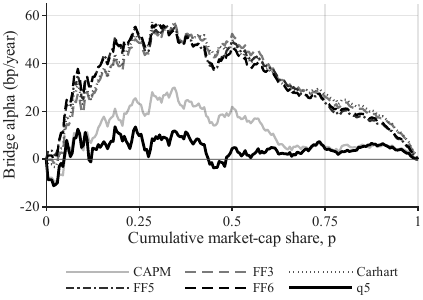}
		\caption{Annual Jul--Jun rebalancing, monthly frequency: cap-axis bridge-alpha curve}
		\label{fig:monthly-bridge}
		\begin{minipage}{\textwidth}
			\footnotesize
			\emph{Note:} Same format as \Cref{fig:daily-bridge} but at the monthly factor frequency. The FF3, FF5, FF6, and Carhart curves form a positive one-signed bulge reaching a maximum of about $+55$ to $+57$\,bp near $p\approx0.33$. The q5 curve stays closer to zero and changes sign, with $|\widehat{SA}|/\widehat{IAE}=0.83$. The sample and estimation procedure are as in \Cref{tab:monthly-bridge}.
		\end{minipage}
	\end{figure}
	
	\paragraph{Monthly attenuation of q5.}
	At the monthly frequency, q5's large negative daily signed area is nearly absent. The signed area is \(+4.4\) bp, the HAC statistic is \(0.40\), and \(p^{NW}_{SA}=0.690\). The nonlinear functionals also do not reject the zero bridge-alpha null: \(p^{HAC\text{-}GP}_{IAE}=0.870\), \(p^{HAC\text{-}GP}_{ISE}=0.910\), and \(p^{HAC\text{-}GP}_{SUP}=0.958\). Thus the directional daily q5 bridge does not persist at the monthly factor frequency.
	
	\paragraph{Monthly Fama--French and Carhart bridge.}
	The Fama--French and Carhart models display the opposite pattern. FF3, Carhart, FF5, and FF6 all have signed areas around \(+32\) bp, coherence ratios near one, HAC statistics between \(2.60\) and \(2.87\), and HAC-Gaussian-process functional \(p\)-values below 0.05. The cap-axis distortion that is weak or absent for these models daily becomes visible at the monthly factor frequency. CAPM has a small positive curve, but it does not reject the formal zero bridge-alpha null.
	
	\paragraph{Frequency interaction.}
	The main empirical message is a frequency-by-model interaction. q5's cap-axis distortion is concentrated at the daily frequency and attenuates monthly, while the Fama--French and Carhart distortions are stronger monthly. This asymmetry matters for interpretation: daily non-synchronous trading and short-horizon beta attenuation are natural candidate explanations for a daily bridge, so the q5 result requires the lead--lag correction in \Cref{subsec:robust-dimson}, but they are less natural explanations for a bridge that becomes clearer at the monthly frequency. The diagnostic therefore does not produce a single model ranking; it shows where, in what direction, and at which frequency cap-rank pricing errors appear.
	
	\FloatBarrier
	
	\section{The Cap--Axis Metric as a Zero-Alpha Factor Coordinate}
	\label{sec:coordinate}
	
	The preceding sections estimate cap-axis pricing errors at the model level. This section uses the same object as a factor-level coordinate. For each candidate factor, I compare its maximum-Sharpe contribution with the cap-axis footprint left by the two-factor model that adds the factor to the market. The resulting map shows that mean--variance contribution and cap-rank zero-alpha consistency order factors differently. Across 155 factors, their rank association is slightly negative, while the positive linear association is driven primarily by a single upper-tail observation. The cap-axis footprint is also only weakly explained by conventional size-factor exposure.
	
	\subsection{Construction}
	\label{subsec:frontier-construction}
	
	For each candidate factor \(g_j\), I estimate a two-factor model containing the market and \(g_j\):
	\begin{equation}
		D_t(p)
		=
		\alpha_j(p)
		+
		\beta_{j,M}(p)R^e_{M,t}
		+
		\beta_{j,g}(p)g_{j,t}
		+
		\varepsilon_{j,t}(p).
		\label{eq:capm-plus-one-frontier}
	\end{equation}
	The factor's zero-alpha footprint is the bridge-alpha curve \(p\mapsto\alpha_j(p)\), summarized by \(\widehat{SA}_j\), \(\widehat{IAE}_j\), \(\widehat{ISE}_j\), and \(\widehat{SUP}_j\). I compare these statistics with the annualized maximum-Sharpe-ratio gain from adding \(g_j\) to the market, denoted \(\Delta SR_j\).
	
	The candidate set is the monthly factor library of \citet{JKP23}, accessed through Global Factor Data \citep{GFD}, augmented with the Fama--French/Carhart and q5 factors used elsewhere in the paper. Requiring complete coverage of the January 1967--December 2024 common sample leaves 155 factors: 145 anomaly-library factors, six Fama--French/Carhart factors, and four q5 nonmarket factors. Every factor is evaluated over the same 696 months, using the monthly-rebalanced cap-axis bridge constructed from the daily path and summed to calendar months.
	
	I read the scatter as a descriptive factor map, not as 155 separate hypothesis tests. The magnitude functionals contain nearly the same ordinal information: the Spearman correlation between \(\widehat{IAE}\) and \(\widehat{ISE}\) is \(0.996\). I use \(\widehat{IAE}\) in the main figure because it is measured in annualized basis points.
	
	\subsection{Sharpe gain and zero-alpha footprint}
	\label{subsec:frontier-sharpe}
	
	\begin{figure}[!htbp]
		\centering
		\includegraphics[width=0.62\linewidth]{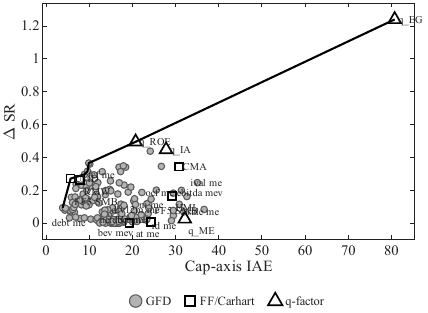}
		\caption{Cap-axis footprint versus Sharpe-ratio gain across the factor universe}
		\label{fig:frontier-iae}
		\begin{minipage}{0.95\textwidth}
			\footnotesize
			\emph{Note:} Each point is one candidate factor added to the market over the common January 1967--December 2024 monthly sample. The horizontal axis is the integrated absolute error \(\widehat{IAE}_j\) of the cap-axis bridge-alpha curve of the two-factor model \(\{R^e_M,g_j\}\), in annualized basis points. The vertical axis is the annualized maximum-Sharpe-ratio gain \(\Delta SR_j\) from adding \(g_j\) to the market. The cap-axis bridge is monthly rebalanced. Open circles are Global Factor Data candidates, open squares are Fama--French/Carhart factors, and filled triangles are q5 factors. The solid line traces the upper-left frontier: among candidate factors with no larger cap-axis footprint, it records the largest observed Sharpe-ratio gain.
		\end{minipage}
	\end{figure}
	
	\Cref{fig:frontier-iae} shows that the cap-axis footprint and the Sharpe coordinate order factors differently. Across the 155 factors, the Spearman correlation between \(\Delta SR\) and \(\widehat{IAE}\) is \(-0.079\). The Pearson correlation is \(0.388\), but this positive linear relation is concentrated in the upper tail. q5 expected growth is the clearest case: it has both the largest Sharpe gain, \(\Delta SR=1.24\), and the largest footprint, \(\widehat{IAE}=80.7\) bp. Excluding this observation reduces the Pearson correlation to \(0.016\), while the Spearman correlation remains slightly negative at \(-0.100\).
	
	The high-Sharpe, low-footprint region is well populated: 45 factors lie above the median \(\Delta SR\) while remaining below the median \(\widehat{IAE}\). The scatter therefore does not reveal a general tradeoff in which greater mean--variance contribution requires greater cap-axis distortion. The narrower conclusion is that the two coordinates are distinct: \(\Delta SR_j\) measures what a factor adds to the attainable frontier, whereas \(\widehat{IAE}_j\) measures the pricing-error curve left on a fixed market-internal subspace when the same factor is added to the market.
	
	\FloatBarrier
	
	\subsection{The footprint is not size exposure}
	\label{subsec:frontier-size}
	
	Because the axis is ordered by market capitalization, a natural concern is that the cap-axis footprint merely relabels exposure to a conventional size factor. I test this using the Fama--French three-factor SMB as the traded size proxy. For every candidate factor \(g_j\), I estimate
	\begin{equation}
		g_{j,t}
		=
		a_j
		+
		b_{j,M}R^e_{M,t}
		+
		\gamma_{j,size}SMB^{FF3}_t
		+
		u_{j,t},
		\label{eq:size-loading-audit}
	\end{equation}
	and ask whether \(\lvert\widehat{\gamma}_{j,size}\rvert\) explains the cross-section of cap-axis footprints.
	
	It does not. Regressing \(\widehat{IAE}\) on the absolute size loading gives \(R^2=0.019\), and the rank correlation between \(\lvert\widehat{\gamma}_{size}\rvert\) and \(\widehat{IAE}\) is \(-0.252\). The explanatory power remains small for the other norms: the size-loading \(R^2\) is \(0.003\) for \(\widehat{ISE}\), \(0.024\) for \(\widehat{SUP}\), and \(0.038\) for \(\lvert\widehat{SA}\rvert\). Residualizing the norms with respect to the absolute FF3 SMB loading leaves their rank relations with \(\Delta SR\) essentially unchanged.
	
	\begin{table}[!htbp]
		\centering
		\singlespacing
		\caption{Size-controlled nonreducibility of the cap-axis footprint}
		\label{tab:size-nonreducibility}
		\footnotesize
		\setlength{\tabcolsep}{4.5pt}
		
		\begin{adjustbox}{max width=\textwidth}
			\begin{tabular}{lrrrrr}
				\toprule
				Norm
				& Pearson with \(\Delta SR\)
				& Spearman with \(\Delta SR\)
				& High-SR/low-error count
				& Size-loading \(R^2\)
				& Size-resid.\ Spearman \\
				\midrule
				\(\widehat{ISE}\)  & \(0.582\) & \(-0.055\) & \(42\) & \(0.003\) & \(-0.058\) \\
				\(\widehat{IAE}\)  & \(0.388\) & \(-0.079\) & \(45\) & \(0.019\) & \(-0.074\) \\
				\(\widehat{SUP}\)  & \(0.338\) & \(-0.099\) & \(44\) & \(0.024\) & \(-0.102\) \\
				\(\lvert\widehat{SA}\rvert\) & \(0.357\) & \(-0.085\) & \(45\) & \(0.038\) & \(-0.081\) \\
				\bottomrule
			\end{tabular}
		\end{adjustbox}
		
		\vspace{0.35em}
		\begin{minipage}{\textwidth}
			\footnotesize
			\emph{Note:} The table reports the 155-factor monthly scan. High-SR/low-error count denotes factors above the median \(\Delta SR\) and below the median reported norm. Size-loading \(R^2\) is the cross-sectional explanatory power of the absolute FF3 SMB loading \(\lvert\widehat{\gamma}_{j,\mathrm{size}}\rvert\) from \Cref{eq:size-loading-audit}. Size-resid.\ Spearman is the Spearman correlation between \(\Delta SR\) and the size-residualized norm. The scan uses a monthly-rebalanced cap-axis bridge over the common 1967--2024 sample.
		\end{minipage}
	\end{table}
	
	\Cref{tab:size-nonreducibility} summarizes the audit. The positive Pearson correlations are substantially larger than the rank correlations, reflecting a small number of joint upper-tail observations rather than a monotone relation, and controlling for FF3 SMB exposure does not restore such a relation. A factor can load strongly on conventional size without producing the largest cap-axis footprint, while a factor with modest size loading can leave a large bridge-alpha curve.
	
	\Cref{tab:frontier-extremes} illustrates this distinction at the factor level. q5 expected growth has the largest footprint, \(\widehat{IAE}=80.7\) bp, despite a market-controlled FF3 SMB loading of only \(-0.24\). By contrast, q5 ME and FF5 SMB load \(0.95\) and \(0.98\) on FF3 SMB but leave much smaller footprints. The smallest-footprint group is also heterogeneous in size exposure: its loadings range from approximately zero for gross profitability to \(0.74\) for price.
	
	\begin{table}[!htbp]
		\centering
		\singlespacing
		\caption{Cap-axis footprint, Sharpe gain, and FF3 SMB loading: selected extremes}
		\label{tab:frontier-extremes}
		\footnotesize
		\setlength{\tabcolsep}{4pt}
		\begin{threeparttable}
			\begin{tabular}{llrrr}
				\toprule
				Factor & Family & \(\Delta SR\) & \(\widehat{IAE}\) (bp) & Size \(\beta\) \\
				\midrule
				\multicolumn{5}{l}{\emph{Panel A. Largest cap-axis footprints}} \\
				q5 EG          & q5  & \(1.24\) & \(80.7\) & \(-0.24\) \\
				corr 1260d     & GFD & \(0.08\) & \(36.6\) & \(0.43\) \\
				seas 6--10na   & GFD & \(0.25\) & \(35.0\) & \(-0.03\) \\
				ncoa gr1a      & GFD & \(0.16\) & \(32.6\) & \(0.14\) \\
				q5 ME          & q5  & \(0.03\) & \(32.1\) & \(0.95\) \\
				sale me        & GFD & \(0.07\) & \(31.9\) & \(0.13\) \\
				nncoa gr1a     & GFD & \(0.20\) & \(31.3\) & \(0.08\) \\
				ival me        & GFD & \(0.12\) & \(30.9\) & \(0.02\) \\
				\addlinespace
				\multicolumn{5}{l}{\emph{Panel B. Smallest cap-axis footprints}} \\
				op at          & GFD & \(0.09\) & \(3.8\) & \(-0.17\) \\
				niq at         & GFD & \(0.08\) & \(4.3\) & \(-0.29\) \\
				ope bel1       & GFD & \(0.09\) & \(4.5\) & \(-0.32\) \\
				ni be          & GFD & \(0.08\) & \(5.4\) & \(-0.46\) \\
				prc            & GFD & \(0.03\) & \(5.4\) & \(0.74\) \\
				gp at          & GFD & \(0.07\) & \(5.5\) & \(-0.01\) \\
				\addlinespace
				\multicolumn{5}{l}{\emph{Memo: size factors}} \\
				FF3 SMB        & FF  & \(0.00\) & \(19.3\) & -- \\
				FF5 SMB        & FF  & \(0.01\) & \(24.2\) & \(0.98\) \\
				\bottomrule
			\end{tabular}
			
			\begin{tablenotes}[flushleft]
				\footnotesize
				\item \emph{Note:} \(\Delta SR\) is the annualized maximum-Sharpe-ratio gain from adding the factor to the market. \(\widehat{IAE}\) is the integrated absolute error of the cap-axis bridge-alpha curve of \(\{R^e_M,g_j\}\), in annualized basis points. Size \(\beta\) is the factor's market-controlled loading on FF3 SMB from \Cref{eq:size-loading-audit}; it is omitted for FF3 SMB because that factor is the benchmark size proxy. Panels A and B report the eight largest and six smallest \(\widehat{IAE}\) observations, respectively.
			\end{tablenotes}
		\end{threeparttable}
	\end{table}
	
	The size-nonreducibility claim concerns curve magnitude and shape, not signed direction. Because \(\widehat{SA}\) preserves directional tilt along the cap axis, it is more closely related to signed size exposure: \(\mathrm{Spearman}(\widehat{SA},\widehat{\gamma}_{size})=0.45\). This moderate relation is expected but does not extend to the magnitude norms. The signed area records whether the curve tilts toward the large- or small-cap side, whereas \(\widehat{IAE}\), \(\widehat{ISE}\), and \(\widehat{SUP}\) measure the magnitude and shape of the remaining pricing-error curve.
	
	\FloatBarrier
	
	\subsection{Model-level cancellation in factor space}
	\label{subsec:frontier-cancellation}
	
	The factor-coordinate map also illustrates model-level cancellation. Under the same monthly-rebalanced implementation used in the scan, the full q5 model has a small cap-axis footprint, with \(\widehat{IAE}=4.4\) bp and \(\widehat{ISE}=29.6\). Yet its constituent factors can have large marginal footprints when added to the market one at a time: expected growth has \(\widehat{IAE}=80.7\) bp, size \(32.1\) bp, investment \(27.8\) bp, and profitability \(20.7\) bp.
	
	This comparison does not decompose the full q5 model. The scan measures each factor's marginal cap-axis footprint when added to the market, not its contribution conditional on the remaining q5 factors. It nevertheless shows that model-level consistency need not arise from cap-axis-neutral components. A factor model can combine individually large marginal footprints and still leave a small assembled bridge-alpha curve through offsetting factor exposures. The cap-axis diagnostic therefore characterizes both individual-factor footprints and the residual geometry of the assembled model, which need not coincide.
	
	\section{Robustness}
	\label{sec:robust}
	
	This section checks whether the frequency-by-model patterns are artifacts of implementation choices or finite-sample behavior. The exercises have distinct roles. A size-and-power calibration evaluates the operating characteristics of the nonlinear functional tests. Rebalancing checks test whether the results depend on the annual Jul--Jun formation cycle. Lead--lag corrections ask whether daily bridge alphas, especially q5's negative daily signed area, are absorbed by short-horizon beta attenuation.
	
	The robustness analysis does not replace the formal inference in the main tables. The main \(p\)-values continue to come from the HAC-Gaussian-process null for the bridge-alpha functionals. The checks below instead ask whether those tests have approximately correct finite-sample size and increasing power, whether the formation convention matters, and whether high-frequency timing effects matter.
	
	\subsection{Finite-sample size and power calibration}
	\label{subsec:robust-size-power}
	
	I evaluate the nonlinear functional tests in a simulation calibrated to the empirical monthly bridge panel. For each model, I hold the factor matrix fixed and resample the centered residual curve in six-month circular blocks, preserving contemporaneous dependence across the 201 cap-axis grid points and short-run time-series dependence. The zero-curve design removes the estimated intercept function. Structured alternatives then add \(\lambda\widehat{\alpha}_m(p)\), where \(\widehat{\alpha}_m(p)\) is model \(m\)'s estimated annual-Jul bridge-alpha curve and \(\lambda\in\{0,0.25,0.50,0.75,1.00\}\). Thus \(\lambda=0\) measures empirical size, while positive values trace power as a model-specific cap-axis pricing-error curve is introduced progressively.
	
	For each simulated panel, I re-estimate the bridge regression and recompute \(\widehat{IAE}\), \(\widehat{ISE}\), and \(\widehat{SUP}\). \Cref{tab:robust-size-power} reports 5\% rejection frequencies averaged across CAPM, FF3, Carhart, FF5, FF6, and q5 and across the three nonlinear functionals. I report the baseline HAC-Gaussian-process critical values together with a block-bootstrap calibration as a finite-sample cross-check.
	
	\begin{table}[!htbp]
		\centering
		\singlespacing
		\caption{Finite-sample size and power of the cap-axis functional tests}
		\label{tab:robust-size-power}
		\footnotesize
		\setlength{\tabcolsep}{6pt}
		\begin{threeparttable}
			\begin{tabular}{lrrrrr}
				\toprule
				Critical values
				& \(\lambda=0\)
				& \(\lambda=0.25\)
				& \(\lambda=0.50\)
				& \(\lambda=0.75\)
				& \(\lambda=1.00\) \\
				\midrule
				HAC-GP          & 4.0\% & 7.3\% & 18.3\% & 36.1\% & 53.9\% \\
				Block bootstrap & 4.9\% & 8.7\% & 20.8\% & 39.2\% & 56.4\% \\
				\bottomrule
			\end{tabular}
			\begin{tablenotes}[flushleft]
				\footnotesize
				\item \emph{Note:} The table reports rejection frequencies averaged across six models and three nonlinear functionals, \(IAE\), \(ISE\), and \(SUP\). The nominal level is 5\%. The case \(\lambda=0\) imposes the zero-curve null; positive values scale each model's estimated cap-axis alpha curve. HAC-GP critical values use 50{,}000 Gaussian draws from the estimated finite-grid HAC covariance process. Block-bootstrap critical values use 4{,}999 resamples with six-month circular blocks; their critical-value and evaluation draws are independent. The calibration uses the annual Jul--Jun bridge over the 1967--2024 monthly sample and 4{,}999 evaluation draws.
			\end{tablenotes}
		\end{threeparttable}
	\end{table}
	
	At \(\lambda=0\), rejection rates are close to the nominal 5\% level: 4.0\% under HAC-GP critical values and 4.9\% under the block bootstrap. Rejection rises monotonically with alternative strength, reaching 53.9\% and 56.4\%, respectively, at \(\lambda=1\). This average masks economically informative heterogeneity: power is high for the pronounced Fama--French and Carhart curves and low for the much smaller CAPM and q5 curves. The two critical-value procedures nevertheless produce similar profiles throughout. The nonlinear bridge functionals are therefore approximately size-controlled under the zero-curve null and respond progressively to structured cap-rank pricing errors.
	
	\FloatBarrier
	
	\subsection{Dependence on the rebalancing cycle}
	\label{subsec:robust-rebalancing}
	
	The main design uses annual Jul--Jun formation. I vary the formation cycle to daily, monthly, quarterly, and annual-Jan schedules. The goal is not to search for the strongest rejection, but to check whether the frequency pattern depends on one rebalancing convention.
	
	The pattern is stable. Under daily factor evaluation, q5 has the most negative signed area across formation cycles, while FF5 and FF6 remain closer to zero. Under monthly factor evaluation, FF3, Carhart, FF5, and FF6 display positive bridges, while q5 remains close to zero. The daily--monthly reversal is therefore not an artifact of annual Jul--Jun formation.
	
	\Cref{fig:rebal-compare} shows the result for q5 and FF5. q5 curves are negative at the daily factor frequency and close to zero monthly. FF5 is small daily and positive monthly. The vertical-axis range differs across panels, so the figure should be read for shape and sign rather than exact cross-panel magnitudes.
	
	\begin{figure}[!htbp]
		\centering
		\begin{subfigure}[t]{0.48\linewidth}
			\centering
			\includegraphics[width=\linewidth]{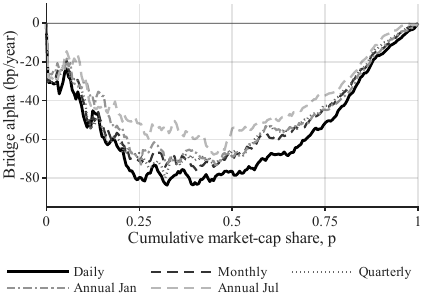}
			\caption{q5, daily frequency}
			\label{fig:rebal-q5-daily}
		\end{subfigure}
		\hfill
		\begin{subfigure}[t]{0.48\linewidth}
			\centering
			\includegraphics[width=\linewidth]{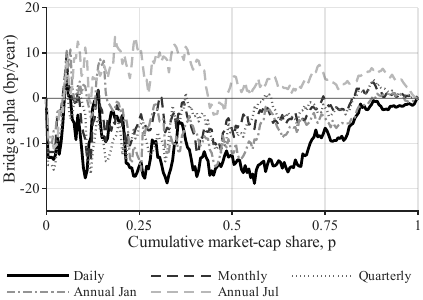}
			\caption{q5, monthly frequency}
			\label{fig:rebal-q5-monthly}
		\end{subfigure}
		
		\vspace{0.6em}
		
		\begin{subfigure}[t]{0.48\linewidth}
			\centering
			\includegraphics[width=\linewidth]{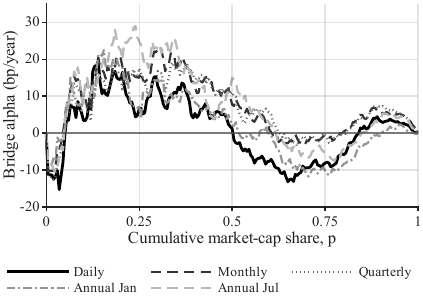}
			\caption{FF5, daily frequency}
			\label{fig:rebal-ff5-daily}
		\end{subfigure}
		\hfill
		\begin{subfigure}[t]{0.48\linewidth}
			\centering
			\includegraphics[width=\linewidth]{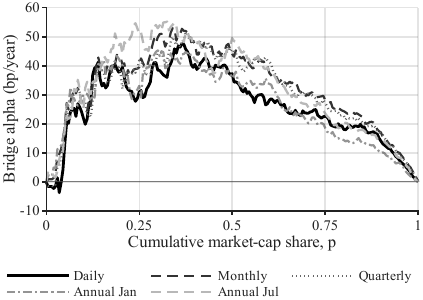}
			\caption{FF5, monthly frequency}
			\label{fig:rebal-ff5-monthly}
		\end{subfigure}
		
		\caption{Cap-axis bridge-alpha curves by formation cycle}
		\label{fig:rebal-compare}
		\begin{minipage}{\linewidth}
			\footnotesize
			\emph{Note:} Each panel plots the bridge-alpha curve, annualized in basis points, against cumulative market-cap share \(p\in[0,1]\), overlaying five formation cycles: daily, monthly, quarterly, annual Jan, and annual Jul. The top row is q5 and the bottom row is FF5. The left column uses daily factors and the right column uses monthly factors. The vertical-axis range differs across panels.
		\end{minipage}
	\end{figure}
	
	Magnitude varies with formation frequency. For q5 under daily factor evaluation, the negative signed area grows as the formation cycle shortens, from \(-39.6\) bp under annual Jul--Jun formation to \(-55.2\) bp under daily formation. This is consistent with a high-frequency component amplified when the cap-rank order and small-stock boundary are updated more often. Alternative grids and HAC lags do not change the ranking of the main functionals.
	
	\FloatBarrier
	
	\subsection{Non-synchronous trading and the lead--lag correction}
	\label{subsec:robust-dimson}
	
	The daily q5 bridge is the result most exposed to a non-synchronous-trading interpretation. If small or illiquid stocks react to factors with a lag, daily betas can be attenuated and bridge alphas can appear. I therefore apply a Dimson-style lead--lag correction, closely related to the Scholes--Williams adjustment.
	
	For each model, I augment the bridge regression with factor leads and lags over a symmetric window \([-k,+k]\). The corrected exposure is the sum of coefficients across the window, and the corrected alpha is the intercept of the extended regression. I report \(k\in\{1,2,3,5,8,10,13,21\}\). The naive estimate is recomputed on each \(k\)-trimmed sample, so changes across columns reflect the correction rather than sample dropout.
	
	\Cref{tab:dimson-gate} reports the horizon profile for the main annual Jul--Jun bridge. The table should be read as a frequency gate, not as a single pass/fail test.
	
	\begin{table}[!htbp]
		\centering
		\singlespacing
		\caption{Horizon profile of the lead--lag-corrected rank-area signed area (annual Jul--Jun rebalancing, ann bp)}
		\label{tab:dimson-gate}
		\footnotesize
		\setlength{\tabcolsep}{3.4pt}
		\begin{threeparttable}
			\begin{tabular}{lrrrrrrrrr}
				\toprule
				Model & Naive & \(k{=}1\) & \(k{=}2\) & \(k{=}3\) & \(k{=}5\) & \(k{=}8\) & \(k{=}10\) & \(k{=}13\) & \(k{=}21\) \\
				\midrule
				CAPM    & \(-10\) & \(1\) & \(5\) & \(8\) & \(13\) & \(15\) & \(15\) & \(17\) & \(20\) \\
				FF3     & \(21\) & \(27^{*}\) & \(27^{*}\) & \(27^{*}\) & \(27^{*}\) & \(26^{*}\) & \(26^{*}\) & \(25^{*}\) & \(22^{*}\) \\
				Carhart & \(21\) & \(28^{*}\) & \(27^{*}\) & \(27^{*}\) & \(27^{*}\) & \(24^{*}\) & \(24^{*}\) & \(24^{*}\) & \(21\) \\
				FF5     & \(7\) & \(18\) & \(20^{*}\) & \(20\) & \(19\) & \(23^{*}\) & \(26^{*}\) & \(28^{*}\) & \(28^{*}\) \\
				FF6     & \(8\) & \(19\) & \(21^{*}\) & \(20\) & \(20\) & \(22^{*}\) & \(24^{*}\) & \(27^{*}\) & \(26^{*}\) \\
				q5      & \(-38^{***}\) & \(-23^{*}\) & \(-15\) & \(-14\) & \(-9\) & \(-6\) & \(-6\) & \(-1\) & \(6\) \\
				\bottomrule
			\end{tabular}
			\begin{tablenotes}[flushleft]
				\footnotesize
				\item \emph{Note:} Each cell is the lead--lag-corrected rank-area signed area, annualized in basis points and rounded to integers. The Naive column is the uncorrected estimate on the same trimmed sample. \(^{*}\), \(^{**}\), and \(^{***}\) denote 5\%, 1\%, and 0.1\% significance by the Newey--West 21-trading-day \(t\)-statistic. \(k\) is the one-sided number of trading days in the symmetric lead--lag window. The sample is annual Jul--Jun rebalancing, 1967--2024 common sample, with 14,556--14,596 trading days depending on the per-\(k\) trim.
			\end{tablenotes}
		\end{threeparttable}
	\end{table}
	
	The q5 signed area attenuates rapidly. The naive estimate is \(-38\) bp and strongly significant. With one lead and lag it falls to \(-23\) bp; from \(k=2\) onward it is no longer significant; by \(k=21\), it is \(+6\) bp. Most attenuation therefore occurs within the first few leads and lags, where a non-synchronous-trading interpretation is most plausible.
	
	The other models have different profiles. FF3 remains positive and significant through \(k=21\), Carhart remains significant through \(k=13\), and FF5 and FF6 become positive at wider windows. The correction therefore does not mechanically shrink every bridge toward zero. It specifically absorbs the negative daily q5 signed area while leaving the Fama--French-family positive bridge intact or stronger.
	
	\paragraph{A daily-rebalancing worst-case check.}
	Daily formation re-sorts the cap-axis order every trading day and updates the small-stock boundary most frequently. It is therefore a worst-case design for non-synchronous-trading exposure. I apply the same lead--lag correction to the daily-rebalancing bridge in \Cref{tab:dimson-gate-daily}.
	
	\begin{table}[!htbp]
		\centering
		\singlespacing
		\caption{Horizon profile of the lead--lag-corrected rank-area signed area (daily rebalancing, ann bp)}
		\label{tab:dimson-gate-daily}
		\footnotesize
		\setlength{\tabcolsep}{3.4pt}
		\begin{threeparttable}
			\begin{tabular}{lrrrrrrrrr}
				\toprule
				Model & Naive & \(k{=}1\) & \(k{=}2\) & \(k{=}3\) & \(k{=}5\) & \(k{=}8\) & \(k{=}10\) & \(k{=}13\) & \(k{=}21\) \\
				\midrule
				CAPM    & \(-13\) & \(-2\) & \(2\) & \(6\) & \(11\) & \(13\) & \(14\) & \(16\) & \(17\) \\
				FF3     & \(18\) & \(25^{*}\) & \(25^{*}\) & \(25^{*}\) & \(25^{*}\) & \(24^{*}\) & \(24^{*}\) & \(22^{*}\) & \(19\) \\
				Carhart & \(9\) & \(15\) & \(16\) & \(16\) & \(16\) & \(13\) & \(13\) & \(12\) & \(9\) \\
				FF5     & \(1\) & \(11\) & \(13\) & \(14\) & \(12\) & \(18\) & \(21\) & \(24^{*}\) & \(24^{*}\) \\
				FF6     & \(-5\) & \(4\) & \(7\) & \(7\) & \(6\) & \(9\) & \(12\) & \(14\) & \(13\) \\
				q5      & \(-54^{***}\) & \(-41^{***}\) & \(-33^{**}\) & \(-30^{**}\) & \(-25^{*}\) & \(-17\) & \(-17\) & \(-11\) & \(-3\) \\
				\bottomrule
			\end{tabular}
			\begin{tablenotes}[flushleft]
				\footnotesize
				\item \emph{Note:} Each cell is the lead--lag-corrected rank-area signed area, annualized in basis points and rounded to integers. The Naive column is the uncorrected estimate on the same trimmed sample. \(^{*}\), \(^{**}\), and \(^{***}\) denote 5\%, 1\%, and 0.1\% significance by the Newey--West 21-trading-day \(t\)-statistic. \(k\) is the one-sided number of trading days in the symmetric lead--lag window. \(k=21\) is about one trading month. The sample is daily rebalancing, 1967--2024 common sample.
			\end{tablenotes}
		\end{threeparttable}
	\end{table}
	
	The daily-rebalancing profile moves in the same direction but starts from a larger negative value. q5's naive signed area is \(-54\) bp, remains significant through \(k=5\), becomes insignificant at wider windows, and reaches \(-3\) bp by \(k=21\). This is the expected pattern in the worst-case design: the distortion is stronger when the cap-rank boundary is updated daily, but it still disappears over a one-month lead--lag window.
	
	Together, the lead--lag exercises localize the daily q5 bridge in a short-horizon window. The same correction does not absorb the Fama--French-family positive bridge, which appears more clearly at the monthly frequency. The robustness evidence therefore supports the paper's main interpretation: the diagnostic measures where along the cap-rank axis, and at which horizon, each model leaves pricing error.
	
	\FloatBarrier
	
	\section{High-Resolution Zero-Alpha Diagnostics versus Size-Bin Tests}
	\label{sec:whybetter}
	
	The cap-axis diagnostic does not replace tests on size-sorted portfolios. A conventional size-bin test asks whether a finite set of fully invested portfolios has nonzero alphas, whereas the cap-axis diagnostic studies a nested family of zero-investment prefix-versus-market bridges. The two procedures use the same economic ordering but impose different test-asset coordinates.
	
	The distinction is resolution. A bin test assigns one alpha to each prechosen portfolio and can therefore compress variation within a bin. The cap-axis diagnostic estimates the path \(p\mapsto\alpha_m(p)\) over cumulative market-capitalization share, and its functionals summarize the direction, magnitude, concentration, and local extremity of the resulting bridge-alpha curve.
	
	This comparison should not be interpreted as a direct comparison between raw size-bin alphas and bridge alphas. A fully invested decile return and an exposure-scaled zero-investment bridge are different test assets. The relevant apples-to-apples comparison is between the full cap-axis bridge and a coarser bridge observed only at size-bin boundaries.
	
	\subsection{Market-value concentration inside size bins}
	\label{subsec:concentration}
	
	The resolution loss is especially severe for NYSE-breakpoint size deciles because market value is concentrated in the largest bin. \Cref{tab:decile-concentration} reports the average share of aggregate investable market capitalization held by each decile from 1967 to 2024. The top decile holds \(60.3\%\) of market value on average, with a median of \(60.7\%\) and a maximum of \(76.4\%\). The top two deciles together hold \(73.7\%\), and the top three hold \(81.4\%\). Most of the market's economic weight is therefore represented by only a few bin-level coordinates.
	
	The dominant bin is also internally concentrated. Although the top decile contains 160 stocks on average, its effective number of constituents, measured by the inverse Herfindahl of within-decile capitalization weights, is approximately 65, and its five largest firms alone account for \(20.3\%\) of the entire market. The top decile is therefore not a homogeneous large-stock block but a steeply weighted submarket with substantial internal cap-rank structure.
	
	\begin{table}[!ht]
		\singlespacing
		\centering
		\footnotesize
		\caption{Concentration of market value across NYSE-breakpoint size deciles}
		\label{tab:decile-concentration}
		\begin{tabular}{lccccc}
			\toprule
			Decile & Mean cap-share & Median & Max & Mean \# stocks & Effective \# \\
			\midrule
			D01 (small) & 1.6\% & 1.6\% & 3.8\% & 1{,}938 & 1{,}328 \\
			D02 & 1.5\% & 1.5\% & 2.8\% & 586 & 563 \\
			D03 & 1.8\% & 1.7\% & 3.1\% & 389 & 381 \\
			D04 & 2.2\% & 2.2\% & 3.7\% & 307 & 302 \\
			D05 & 2.8\% & 2.7\% & 4.6\% & 254 & 250 \\
			D06 & 3.6\% & 3.5\% & 5.4\% & 214 & 211 \\
			D07 & 5.0\% & 5.0\% & 7.2\% & 194 & 190 \\
			D08 & 7.7\% & 7.9\% & 10.0\% & 179 & 174 \\
			D09 & 13.4\% & 13.6\% & 15.8\% & 166 & 158 \\
			D10 (large) & \textbf{60.3\%} & 60.7\% & 76.4\% & 160 & \textbf{65} \\
			\bottomrule
		\end{tabular}
		
		\par\vspace{4pt}
		\begin{minipage}{0.92\textwidth}
			\footnotesize
			\emph{Note:} Each row reports time-series statistics over the 696 months from January 1967 to December 2024. NYSE size breakpoints are applied to the full investable universe. Cap-share is the decile's share of aggregate investable market capitalization. ``Effective \#'' is the inverse Herfindahl index of within-decile capitalization weights. The top decile holds a median \(60.7\%\) of market value and has an effective constituent count of approximately 65 despite containing 160 stocks on average.
		\end{minipage}
	\end{table}
	
	\subsection{A comparable coarse bridge representation}
	\label{subsec:external}
	
	A decile partition can be translated into the same zero-investment bridge units as the cap-axis diagnostic. Let \(\omega_{q,t}\) be size decile \(q\)'s share of aggregate market capitalization and let \(R_{q,t}\) be its value-weighted return. For the set \(\mathcal Q_k\) containing the \(k\) largest deciles, define cumulative exposure
	\begin{equation}
		s^{(10)}_{k,t}
		=
		\sum_{q\in\mathcal Q_k}\omega_{q,t}
		\label{eq:decile_bridge_share}
	\end{equation}
	and the corresponding decile-boundary bridge
	\begin{equation}
		D^{(10)}_{k,t}
		=
		\sum_{q\in\mathcal Q_k}
		\omega_{q,t}R_{q,t}
		-
		s^{(10)}_{k,t}R^M_t.
		\label{eq:decile_boundary_bridge}
	\end{equation}
	This object has the same zero-investment and equal-exposure form as the stock-level cap-axis bridge, but it observes the cumulative bridge only after each complete decile has entered the prefix.
	
	A ten-bin partition therefore supplies nine nontrivial internal bridge boundaries, plus the two zero endpoints. The baseline cap-axis design evaluates 201 target locations. At its time-series mean cap share, the top decile spans approximately \(p\in[0,0.60]\), so the baseline grid evaluates this economically dominant region at more than one hundred target locations. Some adjacent targets can map to the same whole-stock prefix, but the grid still reveals substantially more of the internal path than the single boundary at the end of the top decile.
	
	This formulation also clarifies why raw decile alphas should not be overlaid directly with bridge alphas. A raw decile alpha evaluates a fully invested bin portfolio; Equation \eqref{eq:decile_boundary_bridge} instead converts the decile partition into a coarse version of the same equal-exposure bridge tested by the cap-axis diagnostic.
	
	\subsection{Resolution as a functional diagnostic}
	\label{subsec:resolution}
	
	The gain from the cap-axis design is localization. A coarse boundary test can show that a cumulative large-stock region leaves a nonzero bridge alpha; the full curve additionally shows where within that region the error emerges, whether it changes sign, and whether it later reverses or cancels.
	
	The main results illustrate this distinction. q5's daily bridge reaches its deepest point near \(p\approx0.45\), well inside the region occupied on average by the top size decile. The monthly Fama--French curves peak near \(p\approx0.33\), also inside that dominant region. A decile-boundary bridge observes the cumulative outcome only after the entire top decile has entered the prefix and cannot localize either interior extremum.
	
	High resolution does not mean conducting 201 unrelated pointwise tests. The empirical object is one jointly estimated alpha vector, and inference uses its cross-grid HAC covariance structure. \(SA\), \(IAE\), \(ISE\), and \(SUP\) summarize the path while accounting for dependence across grid locations. Resolution is therefore increased at the construction stage without reducing the analysis to a collection of unadjusted pointwise \(t\)-statistics.
	
	The methodological advantage is not simply a greater tendency to reject. Size-bin tests remain useful for evaluating familiar investable portfolios and can detect economically important violations in their own right. The cap-axis diagnostic answers a complementary question: how an equal-exposure pricing-error bridge evolves as cumulative market capitalization is traversed. Its contribution is to replace a small number of cap-rank boundaries with a functional map of the same economic coordinate.
	
	\section{Discussion}
	\label{sec:discuss}
	
	\paragraph{What the diagnostic identifies.}
	The cap-axis diagnostic measures an intermediate object between aggregate-market fit and full SDF validity. Its test assets are equal-exposure bridges generated by moving through the market in capitalization order. A zero bridge-alpha vector means that the model leaves no alpha on this family; together with a zero aggregate-market alpha, linearity extends the restriction to the finite return space \(\mathcal{V}_{cap,\mathcal{P}}\) generated by the market and the bridges (\Cref{prop:cap_axis_sufficiency}).
	
	This interpretation is narrower than pricing every raw size portfolio. Because the empirical bridge uses whole stocks, its realized exposure \(s_t(p)\) varies as portfolio weights evolve, so zero alpha on the bridge and market does not mechanically imply zero alpha on the unscaled prefix return. The diagnostic instead asks whether each capitalization prefix differs systematically from an equal-exposure position in the whole market.
	
	The functionals distinguish different features of the curve: \(SA\) measures directional tilt; \(IAE\) measures magnitude without sign cancellation; \(ISE\) emphasizes concentrated distortions; and \(SUP\) records the largest local error. The coherence ratio separates predominantly one-signed curves from curves whose positive and negative regions offset.
	
	\paragraph{Functional inference and finite-sample behavior.}
	Even under a zero population curve, estimation noise produces positive \(IAE\), \(ISE\), and \(SUP\). The finite-grid HAC-Gaussian procedure accounts for this by estimating the joint long-run covariance matrix of the bridge-alpha vector and applying the functionals to simulated zero-curve draws, so the resulting \(p\)-values ask whether the observed curve is too large or locally extreme to be attributed to sampling variation.
	
	The size-and-power calibration supports this interpretation. Rejection is close to the nominal 5\% level when the injected alpha curve is zero and rises monotonically as the model-specific curve is introduced, and HAC-GP and independently calibrated residual-block critical values produce similar power profiles. The rank-area portfolio separately supplies a conventional HAC-tested directional alpha, while its small difference from grid-integrated \(SA\) measures whole-stock boundary quantization.
	
	\paragraph{Relation to standard model-comparison criteria.}
	The diagnostic complements joint alpha tests, Hansen--Jagannathan distance, spanning tests, maximum-Sharpe comparisons, and size-portfolio tests. These procedures evaluate pricing errors or mean--variance performance in their chosen asset spaces; the cap-axis diagnostic instead fixes one market-internal ordering and estimates the shape of the model's zero-alpha violation along that coordinate.
	
	A factor can expand the attainable frontier while leaving a nonzero cap-axis curve; conversely, a model can leave a small cap-axis footprint without being the strongest mean--variance model or pricing other anomaly directions. Relative to coarse size portfolios, the gain is localization: nested whole-stock prefixes reveal where model-relative differences accumulate rather than averaging them within a few bins.
	
	\paragraph{Frequency and horizon dependence.}
	q5 leaves a large negative and nearly one-signed bridge at the daily frequency, but the curve attenuates rapidly when factor leads and lags are included and is small monthly, and under the monthly HAC-GP null its \(IAE\), \(ISE\), and \(SUP\) do not reject. The evidence is therefore more consistent with a short-horizon cap-rank imbalance---potentially involving nonsynchronous trading or factor-frequency implementation---than with an unconditional failure of q5.
	
	The Fama--French and Carhart models display the opposite profile. Their positive bridges are weaker daily but become pronounced and nearly one-signed monthly, and lead--lag correction does not absorb them in the same way. The diagnostic therefore identifies where and at which horizon each model leaves cap-axis pricing error rather than producing a stable global ranking.
	
	\paragraph{The factor-coordinate interpretation.}
	The 155-factor scan provides a cross-sectional version of the same distinction. Cap-axis magnitude is not a monotone transformation of maximum-Sharpe gain: some factors improve the frontier while leaving small footprints, whereas others combine large Sharpe gains with large bridge errors, q5 expected growth being the clearest upper-tail example.
	
	The magnitude norms also do not merely reproduce conventional size exposure. FF3 SMB loadings explain little of the cross-sectional variation in \(IAE\), \(ISE\), or \(SUP\), and size-residualization leaves their relation with \(\Delta SR\) essentially unchanged; signed area differs because it preserves directional size tilt. Thus \(\Delta SR\) measures what a factor adds to the frontier, while the cap-axis norms measure the zero-alpha footprint remaining on a fixed bridge family.
	
	\paragraph{Scope and limitations.}
	A statistical nonrejection does not prove that the population curve is zero, and a zero bridge curve does not establish full SDF validity. Pricing errors may remain along value, momentum, profitability, investment, industry, or idiosyncratic directions outside the cap-axis bridge space; conversely, rejection on this axis does not imply that the model lacks mean--variance value elsewhere.
	
	The curve can nevertheless support further applications: its norms could serve as objectives for factor construction, and differences between curves could identify which factor blocks move particular cap-rank regions. The contribution here is the measurement device: aggregate-market fit, maximum-Sharpe improvement, and cap-axis zero-alpha consistency are empirically distinct criteria.
	
	\section{Conclusion}
	\label{sec:conclusion}
	
	This paper develops a cap-axis zero-alpha diagnostic for factor-model evaluation. Each whole-stock capitalization prefix is paired with an equal realized exposure to the aggregate market, producing a closed bridge-alpha curve. Signed area measures direction, while \(IAE\), \(ISE\), and \(SUP\) measure global and local magnitude. Finite-grid HAC-Gaussian inference delivers the functional tests, and residual-block calibration shows near-nominal size and increasing power against structured alternatives.
	
	In the 1967--2024 CRSP market, q5 leaves a one-sided negative daily bridge that attenuates under lead--lag correction and is small monthly, while the Fama--French and Carhart bridges are weaker daily but become positive and nearly one-signed monthly. Cap-rank pricing errors are therefore model-specific and horizon-dependent rather than summarized by a universal ranking.
	
	Across 155 monthly factors, cap-axis magnitude is neither a monotone transformation of maximum-Sharpe gain nor substantially explained by loading on FF3 SMB. Aggregate-market fit, mean--variance improvement, conventional size exposure, and cap-axis zero-alpha consistency are distinct empirical dimensions.
	
	A zero bridge-alpha curve establishes zero-alpha consistency on the implemented bridge family and, when the aggregate-market alpha is also zero, on the finite span generated by the market and the bridges. It does not establish full SDF validity or imply zero alpha for every raw size portfolio. The diagnostic instead makes the direction, magnitude, location, and horizon of market-internal pricing errors observable.
	
	\paragraph{Funding}
	This research did not receive any specific grant from funding agencies in the public, commercial, or not-for-profit sectors.
	
	\paragraph{Declaration of AI usage} 
	During the preparation of this manuscript, the author used ChatGPT (OpenAI) and Claude (Anthropic) for language refinement and structural clarity. All outputs were reviewed and edited by the author, who takes full responsibility for the content.
	
	\paragraph{Declaration of interest}
	The author declares no competing interests.
	
	\newpage
	\singlespacing

\end{document}